\documentclass{revtex4}
\usepackage{amsmath,amsthm}
\usepackage{rotating}
\usepackage{multirow}

\usepackage{graphicx}

\usepackage{amsfonts}
\newtheorem{theorem}{Theorem}[section]
\newtheorem{lemma}[theorem]{Lemma}
\newtheorem{proposition}[theorem]{Proposition}
\newtheorem{cor}[theorem]{Corollary}

\theoremstyle{remark}
\newtheorem{remark}[theorem]{Remark}

\theoremstyle{definition}
\newtheorem{definition}[theorem]{Definition}

\theoremstyle{example}
\newtheorem{example}[theorem]{Example}

\theoremstyle{notation}
\newtheorem{notation}[theorem]{Notation}
\newcommand{\bra}[1]{\langle#1|}
\newcommand{\ket}[1]{|#1\rangle}

\begin{document}

\title{Independence and totalness of subspaces in phase space methods}            
\author{A. Vourdas}
\affiliation{Department of Computer Science,\\
University of Bradford, \\
Bradford BD7 1DP, UK\\}

\begin{abstract}
The concepts of independence and totalness of subspaces are introduced in the context of 
quasi-probability distributions in phase space, for quantum systems with finite-dimensional Hilbert space.
It is shown that due to the non-distributivity of the lattice of subspaces, there are 
various levels of independence, from pairwise independence up to (full) independence. 
Pairwise totalness, totalness and other intermediate concepts are also introduced, which roughly express that the subspaces 
overlap strongly among themselves, and they
cover the full Hilbert space.
A duality between independence and totalness, that involves orthocomplementation (logical NOT operation), is discussed.
Another approach to independence is also studied, using
Rota's formalism on independent partitions of the Hilbert space. This  is used to define informational independence, which is proved to be equivalent to independence.
As an application, the pentagram (used in discussions on contextuality) is analyzed using these concepts.
\end{abstract}
\maketitle

\section{Introduction}

Phase space methods\cite{P1,P2,P3} play an important role in quantum mechanics.
They study various quasi-probability distributions which are the analogues of joint probabilities, 
for non-commuting variables like the position and momentum.
In this paper we consider quantum systems with finite Hilbert space $H(d)$, and study a wide class of such functions $R(i)$
related to projectors in subspaces $H_1,...,H_n$ of $H(d)$.
Special cases are the $Q$-function, the probability distribution in position space, etc.
In this context we study two concepts, independence and totalness. 

Linear independence (which for simplicity we call independence) is a very fundamental concept in the theory of vector spaces, and in other areas like quantum mechanics that depend on it.
A deeper study of this concept led to the subject of matroids\cite{MA1,MA2, MA3}, which defines independence through some axioms, and then defines the concepts of rank and basis.
In a different context, independence has been used within the general framework of the continuous geometries by von Neumann\cite{CG1,CG2}.
Another approach by Rota and collaborators \cite{A0,A1,A2,A3,A4,A5,A6}, defines independent partitions of the Hilbert space.

Ideas from all these areas are incorporated in the present formalism.
We show that there are various levels of independence for the subspaces $\{H_1,...,H_n\}$.
We use indistinguishably the terms independence and disjointness, but for simplicity in most of the paper we use the term independence only.
We also introduce the concept of totalness, which requires strong overlap between the subspaces, and which is dual to the concept of independence.
We show that there are various levels of totalness.
The existence of various levels of independence and totalness, is intimately related to the non-distributive nature of the lattice of subspaces.

Distributivity is a very fundamental property in classical physics and classical (Boolean) logic. For example, 
a student studies a compulsory module $C$, and in addition to that he has to choose one of two optional modules $O_1$ or $O_2$.   
The following statements are equivalent:
\begin{itemize}
\item
He will study the module $C$ and in addition to that module $O_1$ or module $O_2$.
\item
He will study the modules $C$ and $O_1$ or he will study the modules $C$ and $O_2$.
\end{itemize}
The equivalence looks trivial, because distributivity is deeply embedded in our everyday language and the classsical world, which are formally  described with set theory and Boolean algebra.
In Quantum Mechanics distributivity does not hold,  and we need to develop appropriate language that describes this and plays complementary role to non-commutativity.
Concepts which are trivially equivalent in a distributive structure, might become inequivalent in a non-distributive structure.
The various levels of independence (or disjointness), and the various levels of totalness of
sets $\{H_1,...,H_n\}$ of subspaces of $H(d)$, are examples of this.  

More specifically, in this paper:
\begin{itemize}
\item
We introduce the concept of $n$  independent subspaces, which is generalization of  independent vectors.
We show that for $n\ge 3$,  independence is stronger concept than pairwise independence.
This is related to the non-distributive nature of the lattice of subspaces.
We also introduce various intermediate concepts of independence, and the degree of independence.

\item
We introduce total sets of $n$ subspaces, which are extensions of total sets of vectors. 
Totalness means not just covering of the full Hilbert space $H(d)$, but also strong overlap between the subspaces.
We  show that for $n\ge 3$ 
totalness is stronger than pairwise totalness.
This is related to the fact that the lattice of subspaces is non-distributive.
We also introduce various intermediate concepts of totalness, and the degree of totalness.
\item
There is a duality between a set of independent subspaces, and the total set of 
the orthocomplements of these subspaces.
Orthocomplementation (logical NOT operation) transforms independence into totalness.
These ideas are interpreted in terms of measurements with projectors to these subspaces.

\item
A measurement with a projector to a particular subspace, might give the same result for two different states.
For a given measurement, we partition the Hilbert space into sets (blocks) of states, such that this measurement gives the same result for all states in each block
(when the outcome is `yes'). We then introduce the concept of informationally independent measurements, and show that it is equivalent to independence.
This part, links the present work with Rota's formalism on independent partitions\cite{A0,A1,A2,A3,A4,A5,A6}.

\item
Using these concepts we discuss the pentagram, which is used in contextuality\cite{C0,C1,C2,C3,C4,C5,C6,C7,C8}.
The pentagram within a non-contextual hidden variable theory, uses marginals of joint probability distributions.
They are based on the law of total probability, which in turn depends on distributivity.
In non-distributive structures joint probabilities and their marginals are problematic (e.g., joint probabilities
of non-commuting variables). 

\item
Within the full lattice of subspaces which is non-distributive, there are sublattices which are distributive
(e.g., a sublattice generated by commuting subspaces). In these `islands', independence is equivalent to pairwise independence,  
totalness  is equivalent to pairwise totalness, joint probabilities and their marginals are well defined, etc.

\end{itemize}
Overall, the development of such concepts provides a complementary approach to non-commutativity.
Quantum theory is usually described through non-commutativity, and in this paper it is described through non-distributivity.

In section II we introduce within set theory, the concept of disjointness or independence and also the dual concept of totalness. 
In set theory distributivity holds, and there is a single concept of independence and a single concept of totalness.
These two concepts define partitions, which are used 
in the law of total probability, and in defining marginals from probability distributions.

In section III we present briefly the lattice of subspaces. 
We also use projectors to the subspaces $\{H_i\}$ to define a generalized phase space function $R(i)$.

In section IV we introduce various levels of independence, and define the degree of independence. 
In section V we introduce various levels of totalness, and define the degree of totalness.

In section VI we use independent partitions, to define the concept of informationally independent subspaces and measurements.
We also show that informational independence is equivalent to independence.
Weaker concepts of independence (like pairwise independence), are not informationally independent.

As an application of these ideas, we discuss in section VII the pentagram, which is used in discussions on contextuality.
We conclude in section VIII with a discussion of our results.

\section{Disjointness and totalness in set theory}\label{two}

We consider the set of all subsets of  a finite set $\Omega$ (the  powerset $2^{\Omega}$).
In it we define the conjunction (logical AND), disjunction (logical OR), and negation (logical NOT), as the
intersection, union and complement:
\begin{eqnarray}
A_1\wedge A_2=A_1\cap A_2;\;\;\;\;\;A_1\vee A_2=A_1 \cup A_2;\;\;\;\neg A=\Omega \setminus A.
\end{eqnarray}
The powerset $2^{\Omega}$
with these operations is a Boolean algebra.
The corresponding partial order $\prec$ is `subset'.
The smallest element is the empty set ${\cal O}=\emptyset$, and 
the largest element is ${\cal I}=\Omega$. 

\begin{definition}
\mbox{}
\begin{itemize}
\item[(1)]
The subsets $A_1,...,A_n$ of $\Omega$ are  independent or disjoint, if 
\begin{eqnarray}\label{X1}
(A_1\vee...\vee A_{i-1}\vee A_{i+1}\vee...\vee A_n)\wedge A_i=\emptyset.
\end{eqnarray}
for all $i=1,...,n$.
\item[(2)] 
The subsets $A_1,...,A_n$ of $\Omega$ are pairwise independent or pairwise disjoint, if $A_i\wedge A_j=\emptyset$ for all $i,j$.
\item[(3)]
The subsets $A_1,...,A_n$ of $\Omega$ are weakly independent or weakly disjoint, if 
\begin{eqnarray}
A_1\wedge ...\wedge A_n=\emptyset.
\end{eqnarray}

\end{itemize}
\end{definition}
\begin{definition}
\mbox{}
\begin{itemize}
\item[(1)]
The subsets $A_1,...,A_n$ of $\Omega$ form a total set, if
\begin{eqnarray}\label{X10}
(A_1\wedge...\wedge A_{i-1}\wedge A_{i+1}\wedge...\wedge A_n)\vee A_i=\Omega.
\end{eqnarray}
for all $i=1,...,n$:
\item[(2)]
The subsets $A_1,...,A_n$ of $\Omega$ form a pairwise total set, if $A_i\vee A_j=\Omega$ for all $i,j$.
\item[(3)]
The subsets $A_1,...,A_n$ of $\Omega$ form a weakly total set of subsets, if 
\begin{eqnarray}
A_1\vee ...\vee A_n=\Omega.
\end{eqnarray}
\end{itemize}
The subsets $A_1,...,A_n$ form a partition, if they are disjoint (independent) and they also form a weakly total set. 
\end{definition}
\begin{proposition}
In set theory:
\begin{itemize}
\item[(1)]
Independence is equivalent to pairwise independence.
Independence is stronger concept than weak independence (they are equivalent for $n=2$).
\item[(2)]
Totalness is equivalent to pairwise totalness.
Totalness is stronger concept than weak totalness (they are equivalent for $n=2$).
\end{itemize}
\end{proposition}

\begin{proof}
\begin{itemize}
\item[(1)]
Using the distributivity property of set theory, we rewrite Eq.(\ref{X1}) as
\begin{eqnarray}
(A_1\wedge A_i)\vee...\vee (A_{i-1}\wedge A_i)\vee (A_{i+1}\wedge A_i)\vee...\vee (A_n\wedge A_i)=\emptyset.
\end{eqnarray}
This shows that $A_i\wedge A_j=\emptyset$ and therefore independence is equivalent to pairwise independence.
\item[(2)]
Using the distributivity property of set theory, we rewrite Eq.(\ref{X10}) as
\begin{eqnarray}
(A_1\vee A_i)\wedge...\wedge (A_{i-1}\vee A_i)\wedge (A_{i+1}\vee A_i)\wedge...\wedge (A_n\vee A_i)=\Omega.
\end{eqnarray}
This shows that $A_i\vee A_j=\Omega$ and therefore  totalness is equivalent to pairwise totalness.
\end{itemize}
\end{proof}

\begin{proposition}
The $A_1,...,A_n$ are a total set of subsets, if and only if
the $\neg A_1,...,\neg A_n$ are independent subsets of $\Omega$. 
\end{proposition}
\begin{proof}
Using de Morgan's rule, the negation of $A_i\wedge A_j=\emptyset$, gives $\neg A_i\vee \neg A_j=\Omega$.
\end{proof}

\subsection{Marginal distributions: distributivity and the law of total probability}
The marginals of joint probability distributions are based on the law of total probability in Kolmogorov's probability theory.
\begin{proposition}\label{pro1q}
Let $\Omega$ be a set of alternatives, $B_1,...,B_n$ a partition of the set $\Omega$ and $A \subseteq \Omega$. The law of total probability states that
\begin{eqnarray}\label{12q}
p(A)=\sum _ip(A\cap B_i).
\end{eqnarray}
\end{proposition}
\begin{proof} 
 Using the distributivity property of set theory we get
\begin{eqnarray}
A=A\cap \Omega=A\cap (B_1\cup...\cup B_n)=(A\cap B_1)\cup...\cup (A\cap B_n).
\end{eqnarray}
Since $(A\cap B_i)\cap(A\cap B_j)=\emptyset$, we use the additivity property of Kolmogorov probabilities
\begin{eqnarray}
S_1\cap S_2=\emptyset \;\rightarrow\;p(S_1\cup S_2)=p(S_1)+p(S_2),
\end{eqnarray}
we prove Eq.(\ref{12q}).
\end{proof}
Eq.(\ref{12q}) can be used to define marginals of probability distributions.
The ingredients for the law of total probability, are partitions  and distributivity\cite{VJGP}.
Partitions are based on the concepts of disjointness (independence) and also weak totalness.
We will show below that in non-distributive structures, there are various levels of disjointness (independence) and various levels of totalness.
Consequently the relationship between a joint probability distribution and its marginals becomes problematic.

\section{Quantum systems with variables in ${\mathbb Z}(d)$}

We consider a quantum system $\Sigma (d)$ with variables in ${\mathbb Z}(d)$ (the integers modulo $d$), with states in a $d$-dimensional
Hilbert space $H(d)$\cite{vour,vour2}. 
We also consider an orthonormal basis of `position states' which we denote as
$\ket {X;\alpha}$, where the $a \in {\mathbb Z}(d)$, and the $X$ in the notation indicates position states. 
We also consider another orthonormal basis of `momentum states' which  we denote as
$\ket {P;\beta}$, where the $P$ in the notation indicates momentum states. 
They are related to the position states through a finite Fourier transform:
\begin{eqnarray}
\ket {P; \beta}=\frac{1}{\sqrt{d}}\sum _{\alpha}\omega (\alpha \beta)\ket {X;\alpha};\;\;\;\omega(\alpha)=\exp \left(i\frac{2\pi \alpha}{d}\right)
;\;\;\;\alpha,\beta \in {\mathbb Z}(d).
\end{eqnarray}
Displacement operators in the ${\mathbb Z}(d)\times {\mathbb Z}(d)$ phase space of this system, are defined as
\begin{eqnarray}
D(\alpha , \beta)=Z^{\alpha}X^{\beta}\omega (-2^{-1}\alpha \beta);\;\;\;\;
Z=\sum _m\omega (m)\ket{X;m}\bra{X;m};\;\;\;\;
X=\sum _m \ket{X;m+1}\bra{X;m}
\end{eqnarray}
The factor $2^{-1}$ above, is an element of ${\mathbb Z}(d)$, and it exists only for odd $d$. 
The formalism of finite quantum systems, is slightly different in the cases of odd and even $d$.
Below, in the formulas that use the displacement operators, we assume that the dimension $d$ is an odd integer. 

Acting with $D(\alpha , \beta)$ on a `generic' and normalized fiducial vector $\ket {f}$
\begin{eqnarray}
\ket{f}=\sum f_m \ket{X;m};\;\;\;\sum |f_m|^2=1,
\end{eqnarray}
we get the following $d^2$ states which we call coherent states\cite{vour,vour2}
\begin{eqnarray}\label{coh}
\ket{C;\alpha, \beta}=D(\alpha , \beta)\ket{f}=\sum A_m(\alpha, \beta)\ket{X;m};\;\;\;A_m(\alpha, \beta)=\omega (am-2^{-1}\alpha \beta)f_{m-\beta};\;\;\;\;\alpha , \beta \in {\mathbb Z}(d).
\end{eqnarray}
The $C$ in the notation indicates coherent states.
We can write the $A_m(\alpha, \beta)$ as a $d\times d^2$  matrix, with indices $m$ and the pair $(\alpha, \beta)$ written as one index.
Then the requirement of a generic fiducial vector is that the rank of this matrix is $d$.
In this case any $d$ of the $d^2$ coherent states are linearly independent.

The coherent states obey the resolution of the identity:
\begin{eqnarray}\label{coh}
\frac{1}{d}\sum_{\alpha, \beta}\ket{C;\alpha, \beta}\bra{C;\alpha, \beta}={\bf 1}.
\end{eqnarray}

We will use the notation $H(X;\alpha)$, $H(P;\alpha)$, $H(C;\alpha, \beta)$ for the one-dimensional subspaces of $H(d)$
that contain the states $\ket {X;\alpha}$, $\ket {P;\alpha}$, $\ket{C;\alpha, \beta}$, correspondingly. 
We will also use the notation $\Pi[H(X;\alpha)]$, $\Pi[H(P;\alpha)]$, $\Pi[H(C;\alpha, \beta)]$ for the
projectors to these subspaces:
\begin{eqnarray}
&&\Pi[H(X;\alpha)]=\ket {X;\alpha}\bra{X;\alpha};\;\;\;\sum_{\alpha}\Pi[H(X;\alpha)]={\bf 1}\nonumber\\
&&\Pi[H(P;\alpha)]=\ket {P;\alpha}\bra{P;\alpha};\;\;\;\sum_{\alpha}\Pi[H(P;\alpha)]={\bf 1};\nonumber\\
&&\Pi[H(C;\alpha, \beta)]=\ket{C;\alpha, \beta}\bra{C;\alpha, \beta};\;\;\;\frac{1}{d}\sum_{\alpha, \beta }\Pi[H(C;\alpha, \beta)]={\bf 1}.
\end{eqnarray}

\subsection{The lattice ${\cal L}(d)$ of subspaces}
The Birkhoff-von Neumann lattice
of the closed subspaces
of the Hilbert space, with the operations of conjunction, disjunction and complementation, has been studied extensively in the literature \cite{LO1,LO2,LO3,LO4,LO5,LO6}.

We consider the finite-dimensional Hilbert space $H(d)$, describing the system $\Sigma (d)$.
In the set of subspaces of $H(d)$, we define the conjunction (logical AND) and disjunction (logical OR) \cite{la1,la2,la3,la4,la5}:
\begin{eqnarray}
H_1\wedge H_2=H_1\cap H_2;\;\;\;\;\;H_1\vee H_2={\rm span}(H_1 \cup H_2).
\end{eqnarray}
We stress that the logical OR is not just the union, but it contains superpositions of states in the two spaces.
This will lead later to the distinction between  pairwise independence and independence.

The set of subspaces of $H(d)$
with these operations is a lattice, which we denote as ${\cal L}(d)$.
The corresponding partial order $\prec$ is `subspace'.
The smallest element is ${\cal O}=H(0)$ (the zero-dimensional subspace that contains only the zero vector), and 
the largest element is ${\cal I}=H(d)$. 

The lattice ${\cal L}(d)$ is not distributive.
${\cal L}(d)$ is a modular orthocomplemented lattice.
Modularity is a weak version of distributivity,  and is related to independence. Birkhoff discussed the link between matroids 
(which introduce independence in an abstract way) and modular lattices\cite{la1}. 

Modularity states that
\begin{eqnarray}\label{2}
H_1\prec H_3\;\rightarrow\; H_1\vee (H_2\wedge H_3)=(H_1\vee H_2)\wedge H_3.
\end{eqnarray}
Equivalent to this is the following relation which is valid for any $H_1, H_2, H_3$:
\begin{eqnarray}\label{29}
H_1\wedge (H_2\vee H_3)=H_1\wedge [H_3\vee (H_2\wedge (H_1 \vee H_3))].
\end{eqnarray}

Each subspace has an infinite number of complements.
The orthocomplement of $H_1$ is unique, and is another subspace which we denote as 
$H_1^{\perp}$, with the properties
\begin{eqnarray}\label{3}
&&H_1\wedge H_1^{\perp}={\cal O};\;\;\;\;H_1\vee H_1^{\perp}={\cal I}=H(d);\;\;\;\;(H_1^{\perp})^{\perp}=H_1\nonumber\\
&&(H_1\wedge H_2)^{\perp}=H_1^{\perp}\vee H_2^{\perp};\;\;\;\;(H_1\vee H_2)^{\perp}=H_1^{\perp}\wedge H_2^{\perp}\nonumber\\
&&\dim(H_1)+\dim(H_1^{\perp})=d.
\end{eqnarray}

Orthocomplementation is related to logical NOT, in the description of quantum measurements.
We will use the notation $\Pi(H_1)$ for the projector to the subspace $H_1$. Then
\begin{eqnarray}
\Pi^{\perp}(H_1)=\Pi(H_1^{\perp})={\bf 1}-\Pi(H_1). 
\end{eqnarray}
$H_1^{\perp}$ is the null space of $\Pi(H_1)$, and $H_1$ is the null space of $\Pi(H_1^{\perp})$.

A measurement with $\Pi(H_1)$ on a state $\ket{s}$, will give:
\begin{itemize}
\item
 `yes' with probability $p=\bra{s}\Pi(H_1)\ket{s}$, 
in which case the state will collapse into $\frac{1}{\sqrt p}\Pi(H_1)\ket{s}$
\item
`no' with probability $1-p=\bra{s}\Pi(H_1^{\perp})\ket{s}$, in which case the state will collapse 
into $\frac{1}{\sqrt {1-p}}\Pi(H_1^{\perp})\ket{s}$.
\end{itemize}

An important property of modular lattices \cite{la1}, is that
\begin{eqnarray}\label{12}
\dim (H_1\vee H_2)+\dim (H_1\wedge H_2)=\dim (H_1)+\dim (H_2).
\end{eqnarray}

\begin{definition}
$H_1$ commutes with $H_2$ (we denote this as $H_1{\cal C} H_2$) if
\begin{eqnarray}\label{45}
H_1=(H_1\wedge H_2)\vee (H_1\wedge H_2^{\bot})
\end{eqnarray}
\end{definition}
It can be proved that $H_1{\cal C} H_2$ if and only if $[\Pi (H_1), \Pi (H_2)]=0$.
Commutativity of subspaces is equivalent to commutativity of the projectors to these subspaces.

It is easily seen that:
\begin{itemize}
\item
$H_1 \prec H_2$ implies that $H_1{\cal C} H_2$. 
Therefore $H_1$ commutes with $H_1\vee H_2$ and $H_1\wedge H_2$ (for any $H_2$).
\item
Since ${\cal L}(d)$ is an modular lattice,  if $H_1{\cal C} H_2$, then $H_2{\cal C} H_1$ and also $H_1{\cal C} H_2^{\bot}$,
$H_1^{\bot}{\cal C} H_2$.
\item
Every subspace commutes with $H(d)$ and ${\cal O}$.
\item
If $H_1{\cal C} H_2$ and $H_2{\cal C} H_3$, then the $H_1, H_3$ might not commute (transitivity does not hold).
\end{itemize}

Within the lattice ${\cal L}(d)$ which is non-distributive, there are sublattices which are distributive.
For example, any sublattice of ${\cal L}(d)$ generated by commuting subspaces, is distributive. 
In these `islands' results similar to classical physics do hold. For example, the law of the total probability holds, 
joint probability distributions and their marginals are well defined, etc.

\subsection{Quasi-probability distributions}

Below we consider a set $\{H_1,...,H_n\}$ of $n\ge 2$ proper subspaces of $H(d)$ (which might not have the same dimension). 
\begin{notation}
\begin{eqnarray}\label{xcv}
&&{\mathfrak H}_i=\bigvee _{j\ne i}H_j=H_1\vee...\vee H_{i-1}\vee H_{i+1}\vee...\vee H_n\nonumber\\
&&{\mathfrak H}_i^{\perp}=\bigwedge _{j\ne i}H_j^{\perp}=H_1^{\perp}\wedge...\wedge H_{i-1}^{\perp}\wedge H_{i+1}^{\perp}\wedge...\wedge H_n^{\perp}
\end{eqnarray}
Also
\begin{eqnarray}\label{xcv34}
&&{\mathfrak h}_i=\bigwedge _{j\ne i}H_j=H_1\wedge ...\wedge H_{i-1}\wedge H_{i+1}\wedge ...\wedge H_n\nonumber\\
&&{\mathfrak h}_i^{\perp}=\bigvee _{j\ne i}H_j^{\perp}=H_1^{\perp}\vee...\vee H_{i-1}^{\perp}\vee H_{i+1}^{\perp}\vee...\vee H_n^{\perp}
\end{eqnarray}
\end{notation}
Let $\rho$ be a density matrix, and
\begin{eqnarray}\label{AA}
R(i)={\rm Tr}[\rho \Pi(H_i)]\ge 0;\;\;\;i=1,...,n
\end{eqnarray}
For a given $i$, $R(i)$ is the probability that the measurement $\Pi(H_i)$ will give the outcome `yes'.
However the set $\{R(i)|i=1,...,n\}$ is not in general a probability distribution, but 
it can be viewed as a quasi-probability distribution. 
This is related to the lack of independence between the subspaces $H_i$, and we study this in depth
taking into account the non-distributivity of the quantum structure. 
The index $i$ might be a k-tuple $(\alpha_1,...,\alpha _k)$ which takes a finite number of values.

The concepts of independence and totalness underpin this formalism.
Due to the non-distributivity of the quantum structure, both of these concepts are more complex than in set theory 
discussed earlier in section \ref{two}.

Later, in the study of independence and totalness, we will use two more quasi-probability distributions:
\begin{eqnarray}\label{AAA}
{\widetilde R}(i)={\rm Tr}[\rho \Pi({\mathfrak H}_i^{\perp}\wedge H_i)];\;\;\;{\widehat R}(i)={\rm Tr}[\rho  \Pi({\mathfrak h}_i\vee H_i)].
\end{eqnarray}
Since ${\mathfrak H}_i^{\perp}\wedge H_i\prec H_i\prec {\mathfrak h}_i\vee H_i$, it follows that
\begin{eqnarray}
0\le {\widetilde R}(i)\le R(i) \le {\widehat R}(i).
\end{eqnarray}

The
\begin{eqnarray}
{\widetilde R}(i)={\rm Tr}[\rho \Pi({\mathfrak H}_i^{\perp}\wedge H_i)]=
{\rm Tr}[\rho \Pi(H_1^{\perp}\wedge ...\wedge H_{i-1}^{\perp}\wedge H_i \wedge H_{i+1}^{\perp}\wedge ...\wedge H_n^{\perp})],
\end{eqnarray}
involves the part of the space $H_i$ which overlaps with all $H_j^{\perp}$, and therefore it
does not overlap with any of the $H_j$, for $j\ne i$.
The $\Pi({\mathfrak H}_i^{\perp}\wedge H_i)$ commutes with all $\Pi(H_j)$:
\begin{eqnarray}
[\Pi({H}_j),\Pi({\mathfrak H}_i^{\perp}\wedge H_i)]=0
\end{eqnarray}
The ${\widetilde R}(i)$ is the probability that a measurement $\Pi({\mathfrak H}_i^{\perp}\wedge H_i)$ on a system with density matrix $\rho$,
will give `yes'. In this case the state belongs to $H_i$ and it also belongs to all $H_j^{\perp}$, with $j\ne i$.
Therefore a simultaneous measurement with $\Pi(H_j)$ will give `no', if $j\ne i$.  

The
\begin{eqnarray}
{\widehat R}(i)={\rm Tr}[\rho \Pi({\mathfrak h}_i\vee H_i)]=
{\rm Tr}\{\rho \Pi[(H_1\wedge ...\wedge H_{i-1}\wedge H_{i+1}\wedge ...\wedge H_n)\vee H_i]\},
\end{eqnarray}
involves the disjunction of $H_i$ with the overlap of all $H_j$ (with $j\ne i$).
The ${\widehat R}(i)$ is the probability that a measurement $\Pi({\mathfrak h}_i\vee H_i)$ on a system with density matrix $\rho$,
will give `yes'. In this case the state collapses to a superposition of a state in $H_i$ and another state which belongs to all $H_j$ with $j\ne i$.

\begin{example}\label{exam1}
For the $d$ subspaces $H(X;\alpha)$ we get
\begin{eqnarray}\label{E32}
R(\alpha)=\bra{X;\alpha}\rho \ket {X;\alpha};\;\;\;\sum _{\alpha}R(\alpha)=1;\;\;\;\alpha\in {\mathbb Z}(d).
\end{eqnarray} 
This is the probability distribution in the position space. 
\end{example}
\begin{example}\label{exam2}
For the $2d$ subspaces 
\begin{eqnarray}\label{zm}
&&H_i=H(X;i);\;\;\;i=0,...,d-1\nonumber\\
&&H_i=H(P;i-d);\;\;\;i=d,...,2d-1
\end{eqnarray}
we get
\begin{eqnarray}
&&R(i)=\bra{X;i}\rho \ket {X;i};\;\;\;{\rm if}\;i=0,...,d-1\nonumber\\
&&R(i)=\bra{P;i-d}\rho \ket {P;i-d};\;\;\;{\rm if}\;i=d,...,2d-1\nonumber\\
&&\frac{1}{2}\sum _{i}R(i)=1.
\end{eqnarray} 
This distribution consists of both probabilities in position space and probabilities in momentum space.
Although such a distribution is not used in the literature, it is interesting to apply the concepts of this paper, to it.
\end{example}
\begin{example}\label{exam3}
For the $d^2$ subspaces $H(C;\alpha, \beta)$ we get
\begin{eqnarray}
R(\alpha, \beta)=\bra {C;\alpha, \beta}\rho \ket{C; \alpha, \beta};\;\;\;\frac{1}{d}\sum _{\alpha,\beta}R(\alpha, \beta)=1
;\;\;\;\alpha, \beta \in {\mathbb Z}(d),
\end{eqnarray}
Here the index $i$ is the pair $(\alpha, \beta)\in {\mathbb Z}(d)\times {\mathbb Z}(d)$.
$R(\alpha, \beta)$ is the $Q$-function in the ${\mathbb Z}(d)\times {\mathbb Z}(d)$ phase space.
\end{example}

\section{Levels of independence}
\subsection{Independence}

\begin{proposition}\label{pro11}
The subspaces $H_1,...,H_n$ of $H(d)$ are independent, if one of the following statements,
which are equivalent to each other, holds:
\begin{itemize}
\item[(1)]
For all $i=1,...,n$,
\begin{eqnarray}\label{33}
{\mathfrak H_i}\wedge H_i={\cal O}.
\end{eqnarray}
A state cannot belong to both $H_i$ AND to ${\mathfrak H}_i$ (which contains superpositions of states in all $H_j$ with $j\ne i$).

\item[(2)]
For all $i=1,...,n$,
\begin{eqnarray}\label{33A}
[H_1\vee...\vee H_{i-1}]\wedge H_i={\cal O}.
\end{eqnarray}

\item[(3)]
Any $n$ vectors 
$\ket{v_1}\in H_1,...,\ket{v_n}\in H_n$ (one vector from each of the subspaces $H_i$), are independent:
\begin{eqnarray}\label{ind}
\lambda _1\ket{v_1}+...+\lambda _n \ket{v_n}=0\;\;\rightarrow\;\;\lambda _1=...=\lambda _n=0.
\end{eqnarray}
\end{itemize}
\end{proposition}
\begin{proof}
\begin{itemize}
\item[(1)]
We prove that the first two statements are equivalent.
Our proof is related to the one in \cite{CG2}. 

The fact that the first statement implies the second one, is trivial.
We next prove that if
\begin{eqnarray}\label{AS}
A\wedge H_i=(A\vee H_i)\wedge H_{i+1}={\cal O};\;\;\;A=H_1\vee...\vee H_{i-1}
\end{eqnarray}
then
\begin{eqnarray}
(A\vee H_{i+1})\wedge H_{i}={\cal O}.
\end{eqnarray}
We use the identity of Eq.(\ref{29}) with 
\begin{eqnarray}
H_1\;\rightarrow\;H_i;\;\;\;H_2\;\rightarrow\;H_{i+1};\;\;\;H_3\;\rightarrow\;A.
\end{eqnarray}
and we get
\begin{eqnarray}
H_i\wedge (H_{i+1}\vee A)=H_i\wedge [A\vee (H_{i+1}\wedge (H_i \vee A))].
\end{eqnarray}
The assumptions in Eq.(\ref{AS}) show that the right hand side is ${\cal O}$, and therefore the left hand side is ${\cal O}$.

We next use the identity of Eq.(\ref{29}) with 
\begin{eqnarray}
H_1\;\rightarrow\;H_i;\;\;\;H_2\;\rightarrow\;H_{i+2};\;\;\;H_3\;\rightarrow\;B=A\vee H_{i+1}.
\end{eqnarray}
and we get
\begin{eqnarray}
H_i\wedge (H_{i+2}\vee B)=H_i\wedge [B\vee (H_{i+2}\wedge (H_i \vee B))].
\end{eqnarray}
Using the extra assumption $(H_1\vee...\vee H_{i+1})\wedge H_{i+2}={\cal O}$, we prove that the right hand side is ${\cal O}$, and therefore the left hand side is ${\cal O}$.
We continue in this way, and we prove that the second statement implies the first one.
\item[(2)]
We prove that the first and third statements, are equivalent.
We assume that Eq.(\ref{33}) holds, and prove that if
$\lambda _1\ket{v_1}+...+\lambda _n \ket{v_n}=0$, then $\lambda _1=...=\lambda _n=0$.
Indeed, 
\begin{eqnarray}\label{64}
\lambda _1\ket{v_1}+...+\lambda _{i-1} \ket{v_{i-1}}+\lambda _{i+1} \ket{v_{i+1}}+...+\lambda _n\ket{v_n}=-\lambda _i\ket{v_i}
\end{eqnarray}
From this follows that $\lambda _i=0$ because the left hand side belongs to ${\mathfrak H}_i$,
the right hand side to $H_i$ and ${\mathfrak H}_i\wedge H_i={\cal O}$.
Conversely, if $\lambda _1\ket{v_1}+...+\lambda _n \ket{v_n}=0$ implies that $\lambda _1=...=\lambda _n=0$,
then Eq.(\ref{33}) holds, because if  ${\mathfrak H}_i\wedge H_i\ne {\cal O}$ then we have solution to Eq.(\ref{64}) with $\lambda _i \ne 0$.
\end{itemize}
\end{proof}
\begin{proposition}\label{pro23}
\mbox{}
\begin{itemize}
\item[(1)]
If the set $\{H_1,...,H_n\}$ contains independent subspaces, then the subspaces
in any subset (with cardinality at least $2$) are also  independent.
\item[(2)]
If the subspaces $H_1,...,H_n$ of $H(d)$ are independent, then
\begin{eqnarray}\label{111}
\dim (H_1)+...+\dim( H_n)=\dim( H_1\vee...\vee H_n)\le d.
\end{eqnarray}

\end{itemize}
\end{proposition}
\begin{proof}
\mbox{}
\begin{itemize}
\item[(1)]
If Eq.(\ref{33}) holds for the set $\{H_1,...,H_n\}$, then analogous equation holds for any subset of it.
\item[(2)]
The proof is based on Eq.(\ref{12}).
From Eq.(\ref{33}) with $i=1$, it follows that
\begin{eqnarray}
\dim (H_1)+\dim( H_2\vee...\vee H_n)=\dim( H_1\vee...\vee H_n).
\end{eqnarray}
The $H_2,...,H_n$ are independent, and in the same way we prove that
\begin{eqnarray}
\dim (H_2)+\dim( H_3\vee...\vee H_n)=\dim( H_2\vee...\vee H_n).
\end{eqnarray}
These two equations give
\begin{eqnarray}
\dim (H_1)+\dim (H_2)+\dim( H_3\vee...\vee H_n)=\dim( H_1\vee...\vee H_n).
\end{eqnarray}
We continue in the same way and we prove the proposition.
\end{itemize}
\end{proof}

\subsection{Pairwise independence}

\begin{definition}
The subspaces $H_1,...,H_n$ are pairwise independent, if  $H_i\wedge H_j={\cal O}$ for all $i,j$.
\end{definition}
\begin{proposition}\label{pro12}
\mbox{}
\begin{itemize}
\item[(1)]
The non-distributivity of the lattice ${\cal L}(d)$, implies that
independence is stronger concept than pairwise independence.
\item[(2)]
For subspaces within a distributive sublattice of ${\cal L}(d)$, independence is equivalent to pairwise independence.
An example, is when the $H_1,...,H_n$ commute with each other.
\end{itemize}
\end{proposition}
\begin{proof}
\mbox{}
\begin{itemize}
\item[(1)]
In every lattice\cite{la1,la2,la3,la4,la5}
\begin{eqnarray}\label{101}
&&(H_1\wedge H_i)\vee...\vee (H_{i-1}\wedge H_i)\vee (H_{i+1}\wedge H_i)\vee...\vee (H_n\wedge H_i)\nonumber\\
&&\prec [H_1\vee...\vee H_{i-1}\vee H_{i+1}\vee...\vee H_n]\wedge H_i;\;\;\;i=1,...,n.
\end{eqnarray}
Independence implies that the right hand side is ${\cal O}$, and then
the left hand side is ${\cal O}$. This leads to $H_i\wedge H_j={\cal O}$ for all $i,j$, i.e., pairwise independence.
Therefore independence implies independence of every pair of subspaces.

The converse is not true. Pairwise independence implies that the left hand side is ${\cal O}$, but this does not imply that
the right hand side is ${\cal O}$.
Therefore pairwise independence does not imply independence.
\item[(2)]
In distributive lattices Eq.(\ref{101}) becomes equality.  
Therefore within a distributive sublattice of ${\cal L}(d)$, independence is equivalent to pairwise independence.
\end{itemize}
\end{proof}

\subsection{Degree of independence}\label{sec12}
We have seen that pairwise independence is weaker concept than independence.
Between these two concepts, we introduce intermediate concepts which
we quantify with the degree of independence.

\begin{proposition}\label{pro67}
Let $\{H_1,...,H_n\}$ be $n\ge 3$ pairwise independent subspaces, and 
${\mathfrak H}_i$ the subspaces in Eq.(\ref{xcv}).
If $\ket{v_i}\in {\mathfrak H}_i^{\perp}\wedge H_i$, then
\begin{eqnarray}\label{zxc}
\lambda _1\ket{v_1}+...+\lambda _n \ket{v_n}=0\;\;\rightarrow\;\;\lambda _1=...=\lambda _n=0.
\end{eqnarray}
\end{proposition}
\begin{proof}
We assume that
$\lambda _1\ket{v_1}+...+\lambda _n \ket{v_n}=0$.
Then
\begin{eqnarray}\label{64}
\lambda _1\ket{v_1}+...+\lambda _{i-1} \ket{v_{i-1}}+\lambda _{i+1} \ket{v_{i+1}}+...+\lambda _n\ket{v_n}=-\lambda _i\ket{v_i}
\end{eqnarray}
From this follows that $\lambda _i=0$ because the left hand side belongs to 
${\mathfrak H}_i$, the right hand side belongs to ${\mathfrak H}_i^{\perp}\wedge H_i$, and
\begin{eqnarray}
{\mathfrak H}_i\wedge ({\mathfrak H}_i^{\perp}\wedge H_i)={\cal O}.
\end{eqnarray}
This completes the proof.
\end{proof}

It is seen that the independence relation in Eq.(\ref{ind}) is valid here only for vectors in the subspace ${\mathfrak H}_i^{\perp}\wedge H_i$ of $H_i$.
This is the motivation for introducing in Eq.(\ref{AAA}), the quasi-probability distribution ${\widetilde R}(i)$.
The degree of independence compares the subspaces ${\mathfrak H}_i^{\perp}\wedge H_i$ and $H_i$
or equivalently the  ${\widetilde R}(i)$ with $R(i)$.
\begin{definition}
Let $\rho$ be a density matrix.
The matrix for the degree of independence ${\cal A}$, 
and the degree of independence $\eta (\rho)$, 
are given by
\begin{eqnarray}\label{40}
{\cal A}=\frac{1}{n}\sum _i[\Pi({H}_i)-\Pi({\mathfrak H}_i^{\perp}\wedge H_i)];\;\;\;\eta (\rho)=\frac{1}{n}\sum _i [R(i)-{\widetilde R}(i)]=
{\rm Tr}(\rho {\cal A}).
\end{eqnarray}
\end{definition}
Each $\Pi({H}_i)-\Pi({\mathfrak H}_i^{\perp}\wedge H_i)$ is a projector.
As a sum of projectors, ${\cal A}$ is a $d\times d$ positive semidefinite matrix.
The various $\Pi({H}_i)-\Pi({\mathfrak H}_i^{\perp}\wedge H_i)$ do not commute,
and the corresponding ${\rm Tr}\{\rho[ \Pi({H}_i)- \Pi({\mathfrak H}_i^{\perp}\wedge H_i)]\}$
can be measured using different ensembles described by the same density matrix $\rho$.

There are two extreme cases and many intermediate cases:
\begin{itemize}
\item
If ${\mathfrak H}_i^{\perp}\wedge H_i=H_i$ for all $i$, 
then ${\cal A}=0$ and $\eta (\rho)=0$. 
Therefore the $\{H_1,...,H_n\}$ are independent.
In this case proposition \ref{pro67} reduces to proposition \ref{pro11},
and the independence relation in Eq.(\ref{zxc}) holds for all $\ket{v_i}\in H_i$.
This is the strongest form of independence.
\item
If ${\mathfrak H}_i^{\perp}\wedge H_i={\cal O}$ for all $i$, 
\begin{eqnarray}
{\cal A}=\frac{1}{n}\sum _i\Pi({H}_i);\;\;\;
\eta (\rho)=\frac{1}{n}\sum R(i) 
\end{eqnarray}
and the 
$\{H_1,...,H_n\}$ are pairwise independent.
The independence relation in Eq.(\ref{zxc}) does not hold.
This is the weakest form of independence.
\item
Between these two extreme cases, the $\{H_1,...,H_n\}$ are partially independent.
For a given $\rho$, $\eta (\rho)$ takes values in the interval
\begin{eqnarray}\label{678}
0\le \eta (\rho)\le \frac{1}{n}\sum R(i) 
\end{eqnarray}
In this case the independence relation in Eq.(\ref{zxc}) does not hold for all vectors $\ket {v_i}\in H_i$
(it only holds when $\ket{v_i}$ is in the subspace ${\mathfrak H}_i^{\perp}\wedge H_i$ of $H_i$).
\end{itemize}

Proposition \ref{pro12} shows that the non-equivalence of independence and pairwise independence
(which leads to intermediate concepts) is related to the non-distributivity of the lattice ${\cal L}(d)$. In distributive sublattices of ${\cal L}(d)$, independence is equivalent to pairwise independence.

A summary of the various levels of independence is shown in table \ref{t1}.
\begin{table}
\caption{Various levels of independence (or disjointness) and various levels of totalness, for the subspaces
$\{H_1,...,H_n\}$ of $H(d)$.}
\centering
\begin{tabular}{|c|||c|}\hline
pairwise independence: $H_i\wedge H_j={\cal O}$&pairwise totalness: $H_i\vee H_j=H(d)$\\\hline
independence: $H_i\wedge {\mathfrak H}_i={\cal O}$&totalness: $H_i\vee {\mathfrak h}_i=H(d)$\\\hline
weak independence: $H_1 \wedge...\wedge H_n={\cal O}$&weak totalness: $H_1\vee ...\vee H_n=H(d)$\\\hline
matrix for degree of independence: & matrix for the degree of totalness:\\
${\cal A}=\frac{1}{n}\sum [\Pi({H}_i)-\Pi({\mathfrak H}_i^{\perp}\wedge H_i)]$
& 
${\cal T}=\frac{1}{n}\sum  [\Pi({\mathfrak h}_i\vee H_i)-\Pi({H}_i)]$
\\\hline
degree of independence: & degree of totalness:\\
$\eta (\rho)={\rm Tr}(\rho {\cal A})=\frac{1}{n}\sum [R(i)-{\widetilde R}(i)]$
& 
$\epsilon (\rho)={\rm Tr}(\rho {\cal T})=\frac{1}{n}\sum  [{\widehat R}(i)-R(i)]$
\\\hline
\end{tabular}\label{t1}
\end{table}

\begin{example}\label{exA1}
This is related to example \ref{exam1}.
We consider the set $\{H(X; \alpha)\}$ with $d$ subspaces of $H(d)$,  where $\alpha \in {\mathbb Z}(d)$.
These subspaces are pairwise independent.
Also
\begin{eqnarray}
{\mathfrak H}(X; \alpha)=\bigvee _{\beta \ne \alpha} H(X;\beta);\;\;\;
[{\mathfrak H}(X; \alpha)]^{\perp}=H(X; \alpha).
\end{eqnarray} 
and
\begin{eqnarray}
{\cal A}=0;\;\;\;\eta (\rho)=0.
\end{eqnarray}
Therefore these subspaces are independent.
This is the highest level of independence.
In this example ${\widetilde R}(\alpha)=R(\alpha)$ where $R(\alpha)$ has been given in Eq.(\ref{E32}).
\end{example}

\begin{example}\label{exA2}
This is related to example \ref{exam2}.
We consider the set $\{H(X; \alpha), H(P; \beta)\}$ with $2d$ subspaces of $H(d)$, labelled as in Eq.(\ref{zm}). 
These subspaces are pairwise independent ($H_i\wedge H_j={\cal O}$).

In this case 
\begin{eqnarray}
&&i=0,...,(d-1)\;\rightarrow\;{\mathfrak H}_i=\left (\bigvee _{\beta \ne \alpha} H(X;\beta)\right )\vee \left (\bigvee _{\gamma } H(P;\gamma)\right )=H(d);\;\;\;
{\mathfrak H}_i^{\perp}={\cal O}\nonumber\\
&&i=d,...,(2d-1)\;\rightarrow\;{\mathfrak H}_i=\left (\bigvee _{\beta } H(X;\beta)\right )\vee \left (\bigvee _{\gamma \ne \alpha} H(P;\gamma)\right )=H(d);\;\;\;
{\mathfrak H}_i^{\perp}={\cal O}.
\end{eqnarray} 
and
\begin{eqnarray}
{\cal A}=\frac{1}{2d}\sum \Pi({H}_i)=\frac{1}{d}{\bf 1};\;\;\;\eta (\rho)=\frac{1}{d}.
\end{eqnarray}
Therefore these subspaces are pairwise independent.
This is the lowest level of independence.
In this example, ${\widetilde R}(i)=0$.

\end{example}

\begin{example}\label{exA3}
This is related to example \ref{exam3}.
We consider the set $\{H(C; \alpha, \beta)\}$ with $d^2$ subspaces of $H(d)$.
These subspaces are pairwise independent.
We have explained earlier that any $d$ of the $d^2$ coherent states are linearly independent, and therefore
\begin{eqnarray}
&&{\mathfrak H}(C; \alpha _0, \beta _0)=\bigvee _{\alpha \ne \alpha _0, \beta \ne \beta _0} H(C; \alpha, \beta)=H(d);\;\;\;
[{\mathfrak H}(C; \alpha _0, \beta _0)]^{\perp}={\cal O}
\end{eqnarray} 
and
\begin{eqnarray}
{\cal A}=\frac{1}{d^2}\sum \Pi[H(C; \alpha, \beta)]=\frac{1}{d}{\bf 1};\;\;\;\eta (\rho)=\frac{1}{d}.
\end{eqnarray}
Therefore these subspaces are pairwise independent.
This is the lowest level of independence.
In this example, ${\widetilde R}(\alpha, \beta)=0$.
\end{example}

\begin{example}\label{exA4}
In $H(6)$ we consider the following two-dimensional subspaces:
\begin{eqnarray}
H_1=\left \{
\begin{pmatrix}
a\\
b\\
0\\
0\\
0\\
0\\
\end{pmatrix}\right \};\;\;\;
H_2=\left \{
\begin{pmatrix}
0\\
0\\
a\\
b\\
0\\
0\\
\end{pmatrix}\right \};\;\;\;
H_3=\left \{
\begin{pmatrix}
0\\
a\\
0\\
0\\
a\\
b\\
\end{pmatrix}\right \}.
\end{eqnarray}
Here we give a generic vector within these subspaces, which depends on two variables because the
subspaces are two-dimensional.
Then we calculate the subspaces ${\mathfrak H}_i$ (Eq.(\ref{xcv})):
\begin{eqnarray}
{\mathfrak H}_1=H_2\vee H_3=\left \{
\begin{pmatrix}
0\\
a\\
b\\
c\\
a\\
d\\
\end{pmatrix}\right \};\;\;\;
{\mathfrak H}_2=H_1\vee H_3=\left \{
\begin{pmatrix}
a\\
b\\
0\\
0\\
c\\
d\\
\end{pmatrix}\right \};\;\;\;
{\mathfrak H}_3=H_1\vee H_2=\left \{
\begin{pmatrix}
a\\
b\\
c\\
d\\
0\\
0\\
\end{pmatrix}\right \}.
\end{eqnarray}
They are four-dimensional subspaces, and the vectors depend on four variables.

We also consider their orthocomplements which are the subspaces:
\begin{eqnarray}
{\mathfrak H}_1^{\perp}=\left \{
\begin{pmatrix}
a\\
b\\
0\\
0\\
-b\\
0\\
\end{pmatrix}\right \};\;\;\;
{\mathfrak H}_2^{\perp}=\left \{
\begin{pmatrix}
0\\
0\\
a\\
b\\
0\\
0\\
\end{pmatrix}\right \};\;\;\;
{\mathfrak H}_3^{\perp}=\left \{
\begin{pmatrix}
0\\
0\\
0\\
0\\
a\\
b\\
\end{pmatrix}\right \}.
\end{eqnarray}
They are two-dimensional subspaces, and the vectors depend on two variables.
We then calculate the spaces that are used in proposition \ref{pro67}.
\begin{eqnarray}
{\mathfrak H}_1^{\perp}\wedge H_1=\left \{
\begin{pmatrix}
a\\
0\\
0\\
0\\
0\\
0\\
\end{pmatrix}\right \};\;\;\;
{\mathfrak H}_2^{\perp}\wedge H_2=\left \{
\begin{pmatrix}
0\\
0\\
a\\
b\\
0\\
0\\
\end{pmatrix}\right \};\;\;\;
{\mathfrak H}_3^{\perp}\wedge H_3=\left \{
\begin{pmatrix}
0\\
0\\
0\\
0\\
0\\
a\\
\end{pmatrix}\right \}.
\end{eqnarray}

The corresponding projectors are calculated as follows.
Let $a_1,...,a_k$ be $k$  independent vectors, and $M$ the $d\times k$ matrix $(a_1,...,a_k)$ which has as columns these vectors.
The projector to the space spanned by these $k$ vectors is
\begin{eqnarray}\label{proj}
\Pi=M(M^\dagger M)^{-1}M^\dagger.
\end{eqnarray}

Using this we calculated the matrix for the degree of independence:
\begin{eqnarray}
{\cal A}=\frac{1}{6}
\begin{pmatrix}
0&0&0&0&0&0\\
0&3&0&0&1&0\\
0&0&0&0&0&0\\
0&0&0&0&0&0\\
0&1&0&0&1&0\\
0&0&0&0&0&0\\
\end{pmatrix}.
\end{eqnarray}

We also calculated the distributions $R(i)$, ${\widetilde R}(i)$ and the degree of independence.
We consider an orthonormal basis $\ket {k}$ where $k\in {\mathbb Z}(6)$, and the density matrix
\begin{eqnarray}
\rho=\ket{s}\bra {s};\;\;\;\ket{s}=\frac{1}{\sqrt 7}(\ket {0}+\ket{1}+2\ket{3}+\ket{5}).
\end{eqnarray}
We found that 
\begin{eqnarray}
R(1)=0.285;\;\;\;R(2)=0.571;\;\;\;R(3)=0.214\nonumber\\
{\widetilde R}(1)=0.142;\;\;\;{\widetilde R}(2)=0.571;\;\;\;{\widetilde R}(3)=0.142 
\end{eqnarray}
Therefore
\begin{itemize}
\item
Measurement with $\Pi(H_1)$ will give `yes' with probability $R(1)=0.285$.
\item
Measurement with $\Pi(H_1\wedge {\mathfrak H}_1^{\perp})=\Pi(H_1\wedge H_2^{\perp}\wedge H_3^{\perp})$ 
will give `yes' with probability ${\widetilde R}(1)=0.142$. 
In this case a simultaneous measurement with $\Pi(H_2)$
(which commutes with $\Pi(H_1\wedge H_2^{\perp}\wedge H_3^{\perp})$ ) will give `no'.
\end{itemize}
The result `yes' in the first measurement, means that the system collapses to a state that belongs to $H_1$.
The result `yes' in the second measurement, means that the system collapses to a state that belongs to $H_1$ and also to $H_2^{\perp}$
and also $H_3^{\perp}$ (therefore it does not belong to $H_2$ and it does not belong to $H_3$).
Analogous comments can be made for the other $R(i)$ and ${\widetilde R}(i)$.

The degree of independence is $\eta (\rho)=0.071$.
Therefore we have an intermediate level of independence. 

In this example, the independence relation in Eq.(\ref{zxc}) is valid for:
\begin{eqnarray}
\ket{v_1}\in {\mathfrak H}_1^{\perp}\wedge H_1\prec H_1;\;\;\;
\ket{v_2}\in {\mathfrak H}_2^{\perp}\wedge H_2=H_2;\;\;\;\ket{v_3}\in {\mathfrak H}_3^{\perp}\wedge H_3\prec H_3.
\end{eqnarray}
It is seen that there are vectors in $\ket{v_1}\in H_1$ and $\ket {v_3}\in H_3$ for which the   implication in Eq.(\ref{zxc}) is not valid.
\end{example}

\subsection{The partial preorder of the various levels of independence}

\begin{definition}
In $H(d)$ (with fixed $d$), we consider various sets of subspaces $S_1=\{H_1,...,H_n\}$, $S_2=\{H_1^{\prime},...,H_m^{\prime}\}$, etc,
with matrices for the degree of independence ${\cal A}_1$, ${\cal A}_2$, etc.
The set of subspaces $S_1$ is more independent than $S_2$
(we denote this as $S_1\sqsupset S_2$), if ${\cal A}_1-{\cal A}_2$ is a negative semidefinite matrix (denoted as ${\cal A}_1-{\cal A}_2\le 0$).
In this case $\eta _1(\rho)\le \eta _2(\rho)$ for all density matrices $\rho$.
\end{definition}
\begin{proposition}
$\sqsupset$ is a partial preorder.
\end{proposition}
\begin{proof}
We consider the following properties:
\begin{itemize}
\item
Reflexivity: $S_1\sqsupset S_1$. This holds, because ${\cal A}_1-{\cal A}_1=0$ is a negative semidefinite matrix.
\item
Transitivity: if $S_1\sqsupset S_2$ and $S_2\sqsupset S_3$ then $S_1\sqsupset S_3$.
This holds because if ${\cal A}_1-{\cal A}_2$ and ${\cal A}_2-{\cal A}_3$ are 
negative semidefinite matrices, then ${\cal A}_1-{\cal A}_3$ is a negative semidefinite matrix.
\item  
Antisymmetry: if $S_1\sqsupset S_2$ and $S_2\sqsupset S_1$ then $S_1=S_2$.
This does not hold. If ${\cal A}_1-{\cal A}_2$ and ${\cal A}_2-{\cal A}_1$ are
negative semidefinite matrices, then ${\cal A}_1={\cal A}_2$, but this does not imply $S_1=S_2$.
\end{itemize}
Since the first two properties hold, but not the last one, the $\sqsupset$ is a partial preorder.
\end{proof}
\subsection{Weakly independent subspaces}\label{sec49}

We introduce the concept of weakly independent subspaces, which is dual through orthocomplementation, to a weakly total set of subspaces introduced later.
\begin{definition}\label{def21}
The subspaces $H_1,...,H_n$ of $H(d)$ are weakly independent, if 
$H_1\wedge ...\wedge H_n={\cal O}$.
\end{definition}
\begin{proposition}
An independent set of subspaces, is also weakly independent set of subspaces. The converse is true when $n=2$, but it is not true when $n\ge 3$.
\end{proposition}
\begin{proof}
The relation
\begin{eqnarray}\label{Z2}
[H_1\wedge...\wedge H_{i-1}\wedge H_{i+1}\wedge...\wedge H_n] \prec [H_1\vee...\vee H_{i-1}\vee H_{i+1}\vee...\vee H_n].
\end{eqnarray}
proves that
\begin{eqnarray}
H_1\wedge...\wedge H_n\prec [H_1\vee...\vee H_{i-1}\vee H_{i+1}\vee...\vee H_n]\wedge H_i.
\end{eqnarray}
For independent subspaces, the right hand side is zero, and therefore the left hand side is zero.
In Eq.(\ref{Z2}), we have equality when $n=2$, and inequality when $n\ge 3$.
This means that the converse is true when $n=2$, but it is not true when $n\ge 3$.
\end{proof}

\section{Levels of totalness}

In this section we introduce the concept of totalness, which is dual (through orthocomplementation) to independence.

\subsection{Total sets of subspaces}

\begin{proposition}\label{def12}
The subspaces $H_1,...,H_n$ of $H(d)$ are a total set, 
if one of the following statements,
which are equivalent to each other, holds:
\begin{itemize}
\item[(1)]
for all $i=1,...,n$:
\begin{eqnarray}\label{133}
{\mathfrak h}_i\vee H_i=H(d),
\end{eqnarray}
where ${\mathfrak h}_i$ has been defined in Eq.(\ref{xcv34}).
Every vector $\ket{v}\in H(d)$, can be written (not uniquely) as a superposition
of a vector in $H_i$, and a vector which is common in all other subspaces $H_j$ with $j\ne i$:
\begin{eqnarray}\label{29A}
\ket{v}=\lambda _i\ket {a_i}+\mu _i\ket {b_i};\;\;\;\ket{a_i}\in H_i;\;\;\;\ket{b_i}\in {\mathfrak h}_i.
\end{eqnarray}
\item[(2)]
For all $i=1,...,n$,
\begin{eqnarray}
[H_1\wedge...\wedge H_{i-1}]\vee H_i=H(d).
\end{eqnarray}

\item[(3)]
 for all $i=1,...,n$, there is no vector in $H(d)$, which is perpendicular to both
${\mathfrak h}_i$ and $H_i$ . 
In other words, if 
\begin{eqnarray}\label{14}
\langle u\ket{v}=0;\;\;\;\ket{v}=\lambda _i\ket {a_i}+\mu _i\ket {b_i};\;\;\;,
\end{eqnarray}
for all $\ket{a_i}\in H_i$, and all
$\ket{b_i}\in {\mathfrak h}_i$,
then $\ket {u}$ is the zero vector.
\end{itemize}
\end{proposition}
\begin{proof}
\begin{itemize}
\item[(1)]
Orthocomplementation of Eqs(\ref{33}), (\ref{33A}), proves the equivalence of the first two parts. 
\item[(2)]
We assume that the first part of the proposition holds, and prove that the third part also holds.
Taking the orthocomplement of both sides in Eq.(\ref{133}), we get 
${\mathfrak h}_i^{\perp}\wedge H_i^{\perp}={\cal O}$.
This shows that there is no vector perpendicular to both
${\mathfrak h}_i$ and $H_i$ . 

Conversely, if the third part of the proposition holds, then 
${\mathfrak h}_i^{\perp}\wedge H_i^{\perp}={\cal O}$,
and by taking the orthocomplement we get Eq.(\ref{133}).
\end{itemize}
\end{proof}
\begin{remark}
Eq.(\ref{133}) shows that in a total set of subspaces, there is a strong overlap between the subspaces:
\begin{eqnarray}
{\mathfrak h}_i=H_1\wedge...\wedge H_{i-1}\wedge H_{i+1}\wedge...\wedge H_n\ne {\cal O};\;\;\;i=1,...,n.
\end{eqnarray}
Any $d-1$ of the subspaces have vectors in common.
\end{remark}
\begin{proposition}\label{pro39}
The $H_1,...,H_n$ are a total set of subspaces, if and only if
the $H_1^{\perp},...,H_n^{\perp}$ are  independent subspaces of $H(d)$. 
\end{proposition}
\begin{proof}
Independence is defined in Eq.(\ref{33}), which we express in terms of the $H_1^{\perp},...,H_n^{\perp}$ as:
\begin{eqnarray}
{\mathfrak h}_i^{\perp}\wedge H_i^{\perp}={\cal O}.
\end{eqnarray}
From this follows that
\begin{eqnarray}
{\mathfrak h}_i\vee H_i=H(d).
\end{eqnarray}
Therefore, according to the definition \ref{def12}, the $H_1,...,H_n$ are a total set of subspaces. 
This argument also holds in the opposite direction, and proves that the converse is true.
\end{proof}
It is seen that orthocomplementation (the logical NOT operation) converts independence into totalness.
The following proposition is dual to proposition \ref{pro23}.
\begin{proposition}\label{pro56}
\mbox{}
\begin{itemize}
\item[(1)]
If the $\{H_1,...,H_n\}$ is a total set of subspaces, then
any subset (with cardinality at least $2$) is also a total set of subspaces.
\item[(2)]
If $\{H_1,...,H_n\}$ is a total set, then
\begin{eqnarray}
\dim (H_1)+...+\dim( H_n)=\dim( H_1\wedge...\wedge H_n)+(n-1)d
\end{eqnarray}
\end{itemize}
\begin{proof}
\mbox{}
\begin{itemize}

\item[(1)]
If the $S_1=\{H_1,...,H_n\}$ is a total set of subspaces, then the
$S_2=\{H_1^{\perp},...,H_n^{\perp}\}$ is a set of independent subspaces.
According to proposition \ref{pro23}, any subset of $S_2$ is a set of independent subspaces,
and consequently the corresponding subset of $S_1$, that contains the orthocomplements of these subspaces, 
is a total set of subspaces.

For an alternative direct proof we consider, as an example, the subset $\{H_2,...,H_n\}$ of $\{H_1,...,H_n\}$.
If Eq.(\ref{133}) holds, the fact that
\begin{eqnarray}
{\mathfrak h}_i\prec H_2\wedge ...\wedge H_{i-1}\wedge H_{i+1}\wedge...\wedge H_n,
\end{eqnarray}
implies that
\begin{eqnarray}
H(d)={\mathfrak h}_i\vee H_i\prec (H_2\wedge...\wedge H_{i-1}\wedge H_{i+1}\wedge...\wedge H_n)\vee H_i.
\end{eqnarray}
Therefore $(H_2\wedge...\wedge H_{i-1}\wedge H_{i+1}\wedge...\wedge H_n)\vee H_i=H(d)$.
This proves that this particular subset is also a total set of subspaces.
The proof for any other subset is analogous to this.
\item[(2)]
This follows from Eq.(\ref{111}) and the fact that $\dim(H_i)+\dim(H_i^{\perp})=d$.
\end{itemize}
\end{proof}
\end{proposition}

\subsection{Pairwise total subspaces}

\begin{definition}
The subspaces $H_1,...,H_n$ are pairwise total, if  $H_i\vee H_j=H(d)$ for all $i,j$.
\end{definition}

The following proposition is dual to proposition \ref{pro12}.

\begin{proposition}\label{pro13}
\mbox{}
\begin{itemize}
\item[(1)]
The non-distributivity of the lattice ${\cal L}(d)$, implies that
totalness is stronger concept than pairwise totalness.
\item[(2)]
For subspaces in a distributive sublattice of ${\cal L}(d)$, totalness is equivalent to pairwise totalness.
\end{itemize}
\end{proposition}
\begin{proof}
\begin{itemize}
\item[(1)]
In every lattice
\begin{eqnarray}\label{101A}
(H_1\vee H_i)\wedge...\wedge (H_{i-1}\vee H_i)\wedge (H_{i+1}\vee H_i)\wedge...\wedge (H_n\vee H_i)\succ {\mathfrak h}_i\vee H_i;\;\;\;i=1,...,n.
\end{eqnarray}
Totalness implies that the right hand side is $H(d)$, and then
the left hand side is $H(d)$. This leads to $H_i\vee H_j=H(d)$ for all $i,j$, i.e., pairwise totalness.
Therefore totalness implies pairwise totalness.

The converse is not true. Pairwise totalness implies that the left hand side is $H(d)$, but this does not imply that
the right hand side is $H(d)$.
Therefore pairwise totalness does not imply totalness.
\item[(2)]
Within a distributive sublattice of ${\cal L}(d)$, Eq.(\ref{101A}) is equality, and  totalness is equivalent to pairwise 
totalness.
\end{itemize}
\end{proof}

\begin{proposition}\label{pro391}
The $H_1,...,H_n$ are a pairwise total set of subspaces, if and only if
the $H_1^{\perp},...,H_n^{\perp}$ are  pairwise independent subspaces of $H(d)$. 
\end{proposition}
\begin{proof}
Pairwise independence for $H_1^{\perp},...,H_n^{\perp}$ is defined as
\begin{eqnarray}
H_i^{\perp}\wedge H_j^{\perp}={\cal O}.
\end{eqnarray}
Orthocomplementation of this gives
\begin{eqnarray}
H_i\vee H_j=H(d).
\end{eqnarray}
Therefore the $H_1,...,H_n$ are a pairwise total set of subspaces. 
This argument also holds in the opposite direction, and proves that the converse is true.
\end{proof}

\subsection{Degree of totalness}
Pairwise totalness is weaker concept than totalness, and there are intermediate concepts between the two.
As we go from pairwise totalness to totalness, the overlap between the subspaces increases.
\begin{proposition}\label{pro134}
Let $\{H_1,...,H_n\}$ be subspaces which are pairwise total, and 
${\mathfrak h}_i$ the subspaces in Eq.(\ref{xcv34}).
Every vector $\ket{v_i}\in {\mathfrak h}_i \vee H_i$, can be written (not uniquely) as a sum
\begin{eqnarray}\label{479}
&&\ket{v_i}=\lambda _i\ket {a_i}+\mu _i\ket {b_i};\;\;\;\ket{a_i}\in H_i;\;\;\;\ket{b_i}\in {\mathfrak h}_i.
\end{eqnarray}
\end{proposition}
\begin{proof}
This follows from the definition of the disjunction ${\mathfrak h}_i \vee H_i$, which is that a vector in this space can 
be written as in Eq.(\ref{479}).
\end{proof}
It is seen that the totalness relation in Eq.(\ref{29A}) is valid here only for vectors in the subspace ${\mathfrak h}_i \vee H_i$ of $H(d)$.
This is the motivation for introducing in Eq.(\ref{AAA}) the quasi-probability distribution ${\widehat R}(i)$.
The degree of totalness compares the subspaces ${\mathfrak h}_i \vee H_i$ and $H_i$ or equivalently ${\widehat R}(i)$ and $R(i)$.
\begin{definition}
The matrix for the degree of totalness ${\cal T}$, and the degree of totalness $\epsilon (\rho)$, are given by
\begin{eqnarray}
{\cal T}=\frac{1}{n}\sum _i [\Pi({\mathfrak h}_i\vee H_i)-\Pi(H_i)];\;\;\;
\epsilon(\rho)={\rm Tr}(\rho {\cal T})=\frac{1}{n}\sum [{\widehat R}(i)-R(i)]
\end{eqnarray}
\end{definition}
Each $\Pi({\mathfrak h}_i\vee H_i)-\Pi(H_i)$ is a projector.
As a sum of projectors, ${\cal T}$ is a $d\times d$ positive semidefinite matrix.

There are two extreme cases and many intermediate cases:
\begin{itemize}
\item
If ${\mathfrak h}_i \vee H_i=H(d)$ for all $i$, the $\{H_1,...,H_n\}$ are by definition a total set.
In this case 
\begin{eqnarray}
{\cal T}={\bf 1}-\frac{1}{n}\sum _i \Pi(H_i);\;\;\;
\epsilon(\rho)={\rm Tr}(\rho {\cal T})=1-\frac{1}{n}\sum _i R(i).
\end{eqnarray}
Eq.(\ref{29A}) holds for all vectors in $H(d)$, and 
proposition \ref{pro134} reduces to  proposition \ref{def12}.
This is the strongest form of totalness.
\item
If ${\mathfrak h}_i \vee H_i=H_i$ for all $i$, Eq.(\ref{29A}) does not hold. In this case the $\{H_1,...,H_n\}$ are a pairwise total set, 
${\cal T}=0$ and $\epsilon(\rho)=0$.
This is the weakest form of totalness.
\item
Between these two extreme cases, the $\{H_1,...,H_n\}$ are a partially total set, in the sense that Eq.(\ref{29A}) does not hold for all vectors
in $H(d)$.
In this case the degree of totalness takes values in the interval 
\begin{eqnarray}
0\le \epsilon (\rho) \le  1-\frac{1}{n}\sum _i R(i).
\end{eqnarray}
\end{itemize}

A summary of the various levels of totalness is shown in table \ref{t1}.

\begin{example}
We consider the set $\{H(X; \alpha)\}$ with $d$ subspaces of $H(d)$,  where $\alpha \in {\mathbb Z}(d)$, and we get
\begin{eqnarray}
{\mathfrak h}(X; \alpha)=\bigwedge _{\beta \ne \alpha} H(X;\beta)={\cal O}.
\end{eqnarray} 
Therefore ${\mathfrak h}(X; \alpha)\vee H(X; \alpha)=H(X; \alpha)$ and 
\begin{eqnarray}
{\widehat R}(\alpha)=\bra{X;\alpha}\rho \ket {X;\alpha}.
\end{eqnarray} 
Taking into account the results in examples \ref{exam1}, \ref{exA1}, we see that in this case
${\widetilde R}(\alpha)={R}(\alpha)={\widehat R}(\alpha)$.
Therefore ${\cal T}=0$, $\epsilon (\rho)=0$, and the $\{H(X; \alpha)\}$ is a pairwise total set of subspaces.

A different problem is to study the totalness of the orthocomplements $\{[H(X; \alpha)]^{\perp}\}$.
We have seen in example \ref{exA1}, that the $\{H(X; \alpha)\}$ are a set of independent subspaces, and this implies
that their orthocomplements $\{[H(X; \alpha)]^{\perp}\}$ are a total set of subspaces. 
In order to verify this directly, we show that 
\begin{eqnarray}
{\mathfrak h}_{\rm ortho}(X; \alpha)=\bigwedge _{\beta \ne \alpha} [H(X;\beta)]^{\perp}=\left [\bigvee _{\beta \ne \alpha} [H(X;\beta)]\right ]^{\perp}=H(X; \alpha).
\end{eqnarray} 
Consequently
\begin{eqnarray}
[H(X; \alpha)]^{\perp}\vee {\mathfrak h}_{\rm ortho}(X; \alpha)={H}(X; \alpha)^{\perp}\vee H(X; \alpha)=H(d).
\end{eqnarray} 
Therefore 
\begin{eqnarray}
{\cal T}_{\rm ortho}={\bf 1}-\frac{1}{d}\sum \Pi\{H(X;\alpha)^{\perp}\}=
\frac{1}{d}\sum \Pi[H(X;\alpha)]=\frac {1}{d}{\bf 1}
;\;\;\;\epsilon _{\rm ortho}(\rho)=\frac{1}{d}.  
\end{eqnarray}
This confirms that the $\{[H(X; \alpha)]^{\perp}\}$ is a total set of subspaces. 
\end{example}

\begin{example}
We consider the set $\{H(X; \alpha), H(P; \beta)\}$ with $2d$ subspaces of $H(d)$, labelled as in Eq.(\ref{zm}). 
In this case
\begin{eqnarray}
&&i=0,...,(d-1)\;\rightarrow\;{\mathfrak h}_i=\left (\bigwedge _{\beta \ne i} H(X;\beta)\right )\wedge \left (\bigwedge _{\gamma } H(P;\gamma)\right )={\cal O}\nonumber\\
&&i=d,...,(2d-1)\;\rightarrow\;{\mathfrak h}_i=\left (\bigwedge _{\beta } H(X;\beta)\right )\wedge \left (\bigwedge _{\gamma \ne i} H(P;\gamma)\right )={\cal O},
\end{eqnarray} 
and ${\mathfrak h}_i \vee H_i=H_i$. Therefore 
\begin{eqnarray}
{\cal T}=0;\;\;\;\epsilon (\rho)=0.
\end{eqnarray}
It is seen that the $\{H(X; \alpha), H(P; \beta)\}$ is a pairwise total set of subspaces.
Taking into account the results in examples \ref{exam2}, \ref{exA2}, we see that in this case
${\widetilde R}(i)=0$ and ${R}(i)={\widehat R}(i)$.
\end{example}

\begin{example}
We consider the set $\{H(C; \alpha, \beta)\}$ that contains $d^2$ subspaces of $H(d)$, and we get
\begin{eqnarray}
&&{\mathfrak h}(C; \alpha _0, \beta _0)=\bigwedge _{\alpha \ne \alpha _0, \beta \ne \beta _0} H(C; \alpha, \beta)={\cal O},
\end{eqnarray} 
and ${\mathfrak h}(C; \alpha _0, \beta _0)\vee H(C; \alpha, \beta)=H(C; \alpha, \beta)$.
Therefore
\begin{eqnarray}
{\cal T}=0;\;\;\;\epsilon (\rho)=0.
\end{eqnarray}
It is seen that the $\{H(C; \alpha, \beta)\}$ is a pairwise total set of subspaces.
Taking into account the results in examples \ref{exam3}, \ref{exA3}, we see that in this case
${\widetilde R}(\alpha, \beta)=0$ and ${R}(\alpha, \beta)={\widehat R}(\alpha, \beta)$.
\end{example}

\begin{example}
In $H(6)$ we consider the following two-dimensional subspaces:
\begin{eqnarray}
H_1=\left \{
\begin{pmatrix}
a\\
b\\
0\\
0\\
0\\
0\\
\end{pmatrix}\right \};\;\;\;
H_2=\left \{
\begin{pmatrix}
0\\
0\\
a\\
0\\
0\\
b\\
\end{pmatrix}\right \};\;\;\;
H_3=\left \{
\begin{pmatrix}
0\\
a\\
0\\
0\\
0\\
b\\
\end{pmatrix}\right \}.
\end{eqnarray}
In this case
\begin{eqnarray}
{\mathfrak h}_1=H_2\wedge H_3=\left \{
\begin{pmatrix}
0\\
0\\
0\\
0\\
0\\
a\\
\end{pmatrix}\right \};\;\;\;
{\mathfrak h}_2=H_1\wedge H_3=\left \{
\begin{pmatrix}
0\\
a\\
0\\
0\\
0\\
0\\
\end{pmatrix}\right \};\;\;\;
{\mathfrak h}_3=H_1\wedge H_2={\cal O}.
\end{eqnarray}
They are one-dimensional subspaces.
Then
\begin{eqnarray}
{\mathfrak h}_1\vee H_1=\left \{
\begin{pmatrix}
a\\
b\\
0\\
0\\
0\\
c\\
\end{pmatrix}\right \};\;\;\;
{\mathfrak h}_2\vee H_2=\left \{
\begin{pmatrix}
0\\
a\\
b\\
0\\
0\\
c\\
\end{pmatrix}\right \};\;\;\;
{\mathfrak h}_3\vee H_3=H_3.
\end{eqnarray}
We used Eq.(\ref{proj}) to calculate the projectors and we found that
\begin{eqnarray}
{\cal T}=\frac{1}{3}
\begin{pmatrix}
0&0&0&0&0&0\\
0&1&0&0&0&0\\
0&0&0&0&0&0\\
0&0&0&0&0&0\\
0&0&0&0&0&0\\
0&0&0&0&0&1\\
\end{pmatrix}.
\end{eqnarray}

We also calculated the $R(i)$, ${\widehat R}(i)$ and the degree of totalness for the density matrix
\begin{eqnarray}
\rho=\ket{s}\bra {s};\;\;\;\ket{s}=\frac{1}{\sqrt {15}}(\ket {0}+\ket{1}+2\ket{4}+3\ket{5}).
\end{eqnarray}
We found that 
\begin{eqnarray}
R(1)=0.133;\;\;\;R(2)=0.600;\;\;\;R(3)=0.666\nonumber\\
{\widehat R}(1)=0.733;\;\;\;{\widehat R}(2)=0.666;\;\;\;{\widehat R}(3)=0.666 
\end{eqnarray}
Therefore
\begin{itemize}
\item
Measurement with $\Pi(H_1)$ will give `yes' with probability $R(1)=0.133$.
\item
Measurement with $\Pi(H_1\vee {\mathfrak h}_1^{\perp})=\Pi[H_1\vee (H_2\wedge H_3)]$ 
will give `yes' with probability ${\widehat R}(1)=0.733$. 
\end{itemize}
The result `yes' in the first measurement means that the system collapses to a state that belongs to $H_1$.
The result `yes' in the second measurement means that the system collapses to a superposition of a state in $H_1$ and 
another state which belongs to both $H_2$ and $H_3$.
Analogous comments can be made for the other $R(i)$ and ${\widehat R}(i)$.

The degree of totalness is $\epsilon (\rho)=0.222$.
In this example we have an intermediate level of totalness. 

\end{example}

The results for some of the above examples are summarized in table \ref{t2}.

\begin{remark}
Propositions \ref{pro12}, \ref{pro13} show that the following are related:
\begin{itemize}
\item
The lattice ${\cal L}(d)$ is not distributive.
\item
Independence is stronger concept than pairwise independence.

\item
Totalness is stronger concept than pairwise totalness.

\end{itemize}
\end{remark}

\begin{table}
\caption{Some sets of subspaces
$\{H_1,...,H_n\}$ of $H(d)$, and the corresponding  $R(i)$, ${\widetilde R}(i)$, ${\widehat R}(i)$.}
\centering
\begin{tabular}{|c|c|c|}\hline
$\{H(X;\alpha)\}$&$\{H(X;\alpha), H(P;\beta)\}$&$\{H(C; \alpha, \beta)\}$\\\hline
$R(\alpha)=\bra{X;\alpha}\rho\ket{X;\alpha}$&$R(i)=\bra{X;i}\rho\ket{X;i}$&$R(\alpha, \beta)=\bra{C;\alpha, \beta}\rho\ket{C;\alpha, \beta}$\\
&$R(i)=\bra{P;i-d}\rho\ket{P;i-d}$&\\
$\sum R(\alpha)=1$&$\frac{1}{2}\sum R(i)=1$&$\frac{1}{d}\sum R(\alpha, \beta)=1$\\\hline
${\widetilde R}(\alpha)=R(\alpha)$&${\widetilde R}(i)=0$&${\widetilde R}(\alpha, \beta)=0$\\\hline
${\widehat R}(\alpha)=R(\alpha)$&${\widehat R}(i)=R(i)$&${\widehat R}(\alpha, \beta)=R(\alpha, \beta)$\\\hline
\end{tabular}\label{t2}
\end{table}

\subsection{Weakly total sets of subspaces}

We introduce the concept of weak totalness which is dual to the weak independence in section \ref{sec49}.
Weak totalness is the same as 
totalness when $n=2$, and weaker than totalness when $n\ge 3$.

\begin{proposition}
The $H_1,...,H_n$ are a weakly total set of subspaces of $H(d)$, if one of the following statements,
which are equivalent to each other, holds:
\begin{itemize}
\item[(1)]
\begin{eqnarray}\label{13}
H_1\vee ...\vee H_n=H(d).
\end{eqnarray}
Then every vector $\ket{v}\in H(d)$ is a superposition of vectors in $H_i$.
Therefore it can be written (not uniquely) as a sum
\begin{eqnarray}\label{BBB}
\ket{v}=\lambda _1\ket {v_1}+...+\lambda _n\ket {v_n};\;\;\;\ket{v_i}\in H_i
\end{eqnarray}
\item[(2)]
There is no vector in $H(d)$, which is perpendicular to all subspaces $H_1,...,H_n$.
In other words, if 
\begin{eqnarray}\label{14}
\langle a\ket{v}=0;\;\;\;\ket{v}=\lambda _1\ket {v_1}+...+\lambda _n\ket {v_n}
\end{eqnarray}
for all $\ket{v_i} \in H_i$ and all $\lambda _i\in {\mathbb C}$, 
then $\ket {a}$ is the zero vector.
\end{itemize}
\end{proposition}
\begin{proof}
We assume that the first part of the proposition holds, and prove that the second part also holds.
Taking the orthocomplement of both sides in Eq.(\ref{13}), we get $H_1^{\perp}\wedge ...\wedge H_n^{\perp}={\cal O}$.
This shows that there is no vector perpendicular to all subspaces $H_1,...,H_n$, and Eq.(\ref{14}).

Conversely, if the second part of the proposition holds, then $H_1^{\perp}\wedge ...\wedge H_n^{\perp}={\cal O}$.
The orthocomplement of this proves Eq.(\ref{13}), and proves the first part of the proposition.

\end{proof}
\begin{proposition}
If $H_1,...,H_n$ is a weakly total set of independent subspaces of $H(d)$, then
the expansion in Eq.(\ref{BBB}) is unique.
\end{proposition}
\begin{proof}
The $H_1,...,H_n$ are a total set, and therefore there are expansions of an arbitrary vector $\ket{v}$ in $H(d)$, as 
\begin{eqnarray}
\ket{v}=\lambda _1\ket {v_1}+...+\lambda _n\ket {v_n}=\mu _1\ket {v_1}+...+\mu _n\ket {v_n};\;\;\;\ket{v_i}\in H_i.
\end{eqnarray}
This implies that
\begin{eqnarray}
(\lambda _1-\mu _1)\ket {v_1}+...+(\lambda _n-\mu _n)\ket {v_n}=0.
\end{eqnarray}
Since the $H_1,...,H_n$ are independent subspaces, it follows that $\lambda _i=\mu _i$, and this proves the uniqueness of the expansion.
\end{proof}

\begin{proposition}
A total set of subspaces, is also a weakly total set of subspaces. The converse is true when $n=2$, but it is not true when $n\ge 3$.
\end{proposition}
\begin{proof}
By definition of a total set of subspaces ${\mathfrak h}_i\vee H_i=H(d)$. Also
the obvious relation ${\mathfrak h}_i\prec {\mathfrak H}_i$
gives
\begin{eqnarray}
H(d)={\mathfrak h}_i\vee H_i \prec {\mathfrak H}_i\vee H_i =(H_1\vee...\vee H_n).
\end{eqnarray}
Therefore $H_1\vee...\vee H_n=H(d)$.
For $n=2$, we get ${\mathfrak h}_i= {\mathfrak H}_i$.
This means that the converse is true when $n=2$, but it is not true when $n\ge 3$.
\end{proof}

It is seen that when $n\ge 3$, totalness is a stronger concept than weak totalness.
For $n=2$, they are the same.

\begin{proposition}
The $H_1,...,H_n$ are a total set of weakly independent subspaces of $H(d)$, if and only if
the  $H_1^{\perp},...,H_n^{\perp}$ are a weakly total set of independent subspaces of $H(d)$.
In this case the expansion in Eq.(\ref{29A}) is unique.
\end{proposition}
\begin{proof}
\begin{itemize}
\item[(1)]
We have already shown that the $H_1,...,H_n$ are  a total set of subspaces, if and only if the $H_1^{\perp},...,H_n^{\perp}$ are independent.
If in addition to that the $H_1,...,H_n$ are weakly independent, $H_1 \wedge ...\wedge H_n={\cal O}$ and
$H_1^{\perp}\vee ...\vee H_n^{\perp}=H(d)$, which implies that $H_1^{\perp},...,H_n^{\perp}$ are a weakly total set of subspaces of $H(d)$.
The converse is also true.

\item[(2)]
Let
\begin{eqnarray}
\ket{v}=\lambda _i\ket {a_i}+\mu _i\ket {b_i}=\lambda _i^\prime \ket {a_i}+\mu _i^\prime \ket {b_i}
;\;\;\;\ket{a_i}\in H_i
\end{eqnarray}
be two expansions of a vector $\ket{v}$, 
where $\ket{b_i}\in (H_1\wedge...\wedge H_{i-1}\wedge H_{i+1}\wedge...\wedge H_n)$. Then
\begin{eqnarray}
(\lambda _i-\lambda _i^\prime)\ket {a_i}=(\mu _i^\prime-\mu _i)\ket {b_i};\;\;\;\ket{a_i}\in H_i;\;\;\;
\ket{b_i}\in (H_1\wedge...\wedge H_{i-1}\wedge H_{i+1}\wedge...\wedge H_n).
\end{eqnarray}
It follows that $\lambda _i-\lambda _i^\prime=\mu _i^\prime-\mu _i=0$, and therefore 
the expansion is unique.
\end{itemize}
\end{proof}

It is seen that with orthocomplementation, independence and weak independence, become totalness, and weak totalness, correspondingly.

\subsection{Orthogonalization}

In addition to the expansion in Eq.(\ref{BBB}) which involves non-orthogonal components, we can have an orthogonal expansion as discussed below.
\begin{proposition}
Let $H_1,...,H_n$ be a weakly total set of independent subspaces of $H(d)$.
We introduce the following spaces and the corresponding projectors:
\begin{eqnarray}
&&{\cal H}_1=H_1;\;\;\;{\mathfrak P}_1=\Pi(H_1)\nonumber\\
&&{\cal H}_2=(H_1\vee H_2)\wedge H_1^{\perp};\;\;\;{\mathfrak P}_2=\Pi(H_1\vee H_2)-\Pi(H_1)\nonumber\\
&&......\nonumber\\
&&{\cal H}_i=[(H_1\vee...\vee H_{i-1})\vee H_i]\wedge (H_1\vee...\vee H_{i-1})^{\perp}
;\;\;\;{\mathfrak P}_i=\Pi(H_1\vee...\vee H_i)-\Pi(H_1\vee...\vee H_{i-1})
\nonumber\\
&&......\nonumber\\
&&{\cal H}_n=(H_1\vee...\vee H_{n-1})^{\perp};\;\;\;
{\mathfrak P}_n={\bf 1}-\Pi(H_1\vee...\vee H_{n-1})
\end{eqnarray}
Then:
\begin{eqnarray}\label{edc}
{\cal H}_1\vee ...\vee {\cal H}_n=H(d);\;\;\;
{\mathfrak P}_1+...+{\mathfrak P}_n={\bf 1};\;\;\;{\mathfrak P}_i{\mathfrak P}_j={\mathfrak P}_i\delta (i,j).
\end{eqnarray}
\end{proposition}
\begin{proof}
We prove that ${\cal H}_1\vee ...\vee {\cal H}_n=H(d)$ using the modularity property in Eq.(\ref{2}).
For example
\begin{eqnarray}
{\cal H}_1\vee {\cal H}_2=H_1\vee [H_1^{\perp}\wedge (H_1\vee H_2)]=(H_1\vee H_1^{\perp})\wedge (H_1\vee H_2)=H(d)\wedge (H_1\vee H_2)=
H_1\vee H_2.
\end{eqnarray}
We continue in the same way and we prove that
\begin{eqnarray}
{\cal H}_1\vee ...\vee {\cal H}_n=H_1\vee ...\vee H_n=H(d).
\end{eqnarray}

The fact that ${\cal H}_1\vee {\cal H}_2=H_1\vee H_2$ and ${\cal H}_1 \wedge {\cal H}_2={\cal O}$
shows that ${\mathfrak P}_2=\Pi(H_1\vee H_2)-\Pi(H_1)$. In analogous way we prove the expressions given above, for the rest of the projectors.

Direct multiplication proves that they are orthogonal  projectors. For example
\begin{eqnarray}
{\mathfrak P}_1{\mathfrak P}_2=\Pi(H_1)[\Pi(H_1\vee H_2)-\Pi(H_1)]=\Pi(H_1)-\Pi(H_1)=0,
\end{eqnarray}
and
\begin{eqnarray}
{\mathfrak P}_2{\mathfrak P}_2&=&[\Pi(H_1\vee H_2)-\Pi(H_1)][\Pi(H_1\vee H_2)-\Pi(H_1)]\nonumber\\&=&\Pi(H_1\vee H_2)-\Pi(H_1)-\Pi(H_1)+\Pi(H_1)={\mathfrak P}_2.
\end{eqnarray}
This completes the proof.
\end{proof}
Using Eq.(\ref{edc}) we can express an arbitrary state in terms of orthogonal components:
\begin{eqnarray}
\ket{v}=\sum {\mathfrak P}_i\ket{v}=\sum \lambda _i\ket{v_i};\;\;\;|\lambda _i|^2=\bra{v}{\mathfrak P}_i\ket{v};\;\;\;\ket{v_i}\in {\cal H}_i.
\end{eqnarray}
If we change the order of the subspaces, we get different projectors and a different expansion.

\section{Informationally independent subspaces and measurements}

It has been pointed out in a pure mathematics context \cite{A0,A1,A2,A3,A4,A5,A6}, that the lattices describing finite quantum systems (and also the normal subgroups of a group), 
are a special case of modular orthocomplemented lattices, with extra stronger properties.
They are lattices of commuting equivalence relations (also called linear lattices by Rota and collaborators\cite{A2,A3,A4,A5,A6}).
The lattices of commuting equivalence relations are modular, but the converse is not true in general.

Equivalence relations are intimately related to partitions of the Hilbert space $H(d)$, and it is the language of partitions that we use below.
Two partitions are independent, if knowledge of the block of the first partition to which an element belongs, provides no information about  
the block of the second partition to which this element belongs.

Based on the concept of independent partitions, we introduce in this section
informationally independent subspaces, in a physical context. 
Physically, each subspace $H_1$ leads naturally to a partition $\varpi(H_1)$ of the Hilbert space $H(d)$, into blocks which are sets but not subspaces.
Measurement with the projector $\Pi(H_1)$ gives the same result for all states in each block of the 
partition $\varpi(H_1)$ (when the outcome is `yes').

We show that informational independence is equivalent to independence.
Weaker concepts of independence, are not informationally independent.

\subsection{Partitions of the Hilbert space $H(d)$ and their role in quantum measurements}

\begin{definition}
If $H_1$ is a subspace of $H(d)$, 
$\varpi(H_1)$ is the partition of $H(d)$ into `blocks' $\ket{v}+H_1$ that contain
vectors $\ket{v}\in H_1^{\perp}$ modulo vectors $\ket{a}\in H_1$:
\begin{eqnarray}
\ket{v}+H_1=\{\ket{v}+\ket{a}\;|\;\ket{a}\in H_1\};\;\;\;\ket{v}\in H_1^{\perp}.
\end{eqnarray}
The blocks are sets, but they are {\bf not} subspaces.
The partition $\varpi [H(d)]=\varpi ({\cal I})$ has only one block, the $H(d)$.  
In the partition $\varpi({\cal O})$, each block contains one vector $\ket{v}$ only.
\end{definition}

Although it is not essential, it is convinient to consider below normalized vectors:
\begin{eqnarray}
\langle v\ket{v}+\langle a\ket{a}=1;\;\;\;\langle v\ket{a}=0.
\end{eqnarray}

\begin{remark}
\mbox{}
\begin{itemize}
\item
Each normalized vector $\ket{s}$ in $H(d)$, can be written uniquely as
\begin{eqnarray}
\ket{s}=\Pi (H_1^{\perp})\ket{s}+\Pi (H_1)\ket{s}.
\end{eqnarray}
Therefore $\ket{s}$ belongs to exactly one block within the partition $\varpi(H_1)$ (the block $\Pi (H_1^{\perp})\ket{s}+H_1$).
$\ket{s}$ also belongs to exactly one block within the partition $\varpi(H_1^{\perp})$ (the block $\Pi (H_1)\ket{s}+H_1^{\perp}$).

\item
With a measurement $\Pi(H_1^{\perp})$:
\begin{itemize}
\item
The block $\ket{v}+H_1$ in the partition $\varpi (H_1)$, contains states $\ket{s}$ which when the outcome is `yes', collapse into the same state
\begin{eqnarray}\label{63}
\frac{1}{\sqrt p}\ket{v}\in H_1^{\perp};\;\;\;p=\langle v\ket{v}=\bra{s}\Pi (H_1^{\perp})\ket{s}
\end{eqnarray}
with the same probability $p$. 
The measurement $\Pi(H_1^{\perp})$, cannot distinguish the states in the block  $\ket{v}+H_1$, when the outcome is `yes'.
\item
The block $\ket{a}+H_1^{\perp}$ in the partition $\varpi (H_1^{\perp})$, contains states $\ket{s}$ which when the outcome is `no', collapse into the same state
\begin{eqnarray}\label{630}
\frac{1}{\sqrt {1-p}}\ket{a}\in H_1;\;\;\;1-p=\langle a\ket{a}=\bra{s}{\bf 1}-\Pi (H_1^{\perp})\ket{s}
\end{eqnarray}
with the same probability $1-p$. 
The measurement $\Pi(H_1^{\perp})$, cannot distinguish the states in the block  $\ket{a}+H_1^{\perp}$, when the outcome is `no'.
\end{itemize}
\item
Vectors in the two blocks $\ket{v}+H_1$ and $\lambda \ket{v}+H_1$ (with $|\lambda|\le 1$)
with a measurement $\Pi(H_1^{\perp})$ that gives the outcome `yes', will collapse into the same state given in Eq.(\ref{63}),
with different probabilities $p$ and $|\lambda |^2p$, correspondingly.

\end{itemize}
\end{remark}

\begin{proposition}
The partition $\varpi(H_1)$, with the following operation defining superpositions between its blocks
\begin{eqnarray}\label{55}
\lambda_1(\ket{s_1}+H_1)+\lambda _2(\ket{s_2}+H_1)=(\lambda_1\ket{s_1}+\lambda_2\ket{s_2})+H_1;\;\;\;\ket{s_1}, \ket{s_2}\in H_1^{\perp},
\end{eqnarray}
is a Hilbert space isomorphic to $H_1^{\perp}$. The zero vector in $\varpi(H_1)$, is the block $H_1$.
\end{proposition}
\begin{proof}
The notation that we introduced above shows that there is a bijective map between the partition $\varpi(H_1)$ and
the Hilbert space $H_1^{\perp}$. It is also easily seen that the sum of blocks in $\varpi(H_1)$ defined in Eq.(\ref{55})
corresponds to the sum of vectors in $H_1^{\perp}$.
\end{proof}

\begin{cor}
In the set of partitions 
\begin{eqnarray}
\{\varpi (H_1)\;|\;H_1\prec H(d)\}
\end{eqnarray}
we define the following operations:
\begin{itemize}
\item[(1)]
Disjunction
\begin{eqnarray}
\varpi (H_1)\vee \varpi(H_2)=\varpi (H_1\wedge H_2)
\end{eqnarray}
The blocks in $\varpi (H_1)\vee \varpi(H_2)$ are
\begin{eqnarray}
\ket{v}+H_1\wedge H_2;\;\;\;\ket{v}\in (H_1\wedge H_2)^\perp=H_1^{\perp}\vee H_2^{\perp}.
\end{eqnarray}
Special cases are:
\begin{eqnarray}
\varpi (H_1)\vee \varpi ({\cal O})=\varpi ({\cal O});\;\;\;\varpi (H_1)\vee \varpi ({\cal I})=\varpi (H_1);\;\;\;\varpi (H_1)\vee \varpi (H_1^{\perp})=\varpi ({\cal O}).
\end{eqnarray}
\item[(2)]
Conjunction
\begin{eqnarray}
\varpi (H_1)\wedge \varpi(H_2)=\varpi (H_1\vee H_2)
\end{eqnarray}
The blocks in $\varpi (H_1)\wedge \varpi(H_2)$ are
\begin{eqnarray}
\ket{v}+H_1\vee H_2;\;\;\;\ket{v}\in (H_1\vee H_2)^\perp=H_1^{\perp}\wedge H_2^{\perp}.
\end{eqnarray}
Special cases are:
\begin{eqnarray}
\varpi (H_1)\wedge \varpi ({\cal O})=\varpi (H_1);\;\;\;\varpi (H_1)\wedge \varpi ({\cal I})=\varpi ({\cal I});\;\;\;\varpi (H_1)\wedge \varpi (H_1^{\perp})=\varpi ({\cal I}).
\end{eqnarray}
\item[(3)]
Orthocomplement
\begin{eqnarray}
[\varpi (H_1)]^\perp=\varpi (H_1^{\perp})
\end{eqnarray}
The blocks in $[\varpi (H_1)]^\perp$ are
$\ket{v}+H_1^{\perp}$ where $\ket{v}\in H_1$.
Special cases are:
\begin{eqnarray}
[\varpi ({\cal O})]^\perp =\varpi ({\cal I});\;\;\;[\varpi ({\cal I})]^\perp =\varpi ({\cal O}).
\end{eqnarray}
\item[(4)]
The partial order is `refinement'.
$\varpi (H_1)\prec \varpi(H_2)$ if $H_1\succ H_2$, in which case every block in $\varpi(H_2)$ is contained in some block in $\varpi(H_1)$.
For every $\varpi (H_1)$, we get $\varpi({\cal I})\prec \varpi(H_1)\prec \varpi({\cal O})$.
\end{itemize}
Then the set of partitions is a lattice dually isomorphic to ${\cal L}(d)$ (i.e., the conjunction and disjunction exchange roles).
\end{cor}
\begin{proof}
This follows from the fact that $\varpi(H_1^{\perp} )$ is isomorphic to $H_1$, and Eqs.(\ref{3}).
\end{proof}
\begin{definition}
The $\varpi (H_1),...,\varpi(H_n)$ are independent, a total set or a strongly total set, if the
$H_1^{\perp},...,H_n^{\perp}$ are independent, a total set or a strongly total set, correspondingly.
\end{definition}

\subsection{Two informationally independent subspaces and measurements}

Two partitions $\varpi(H_1)$, $\varpi(H_2)$ are informationally independent, if knowledge that a vector belongs to a particular block of the partition $\varpi(H_1)$, gives no information about the block of the partition $\varpi(H_2)$ which contains this vector. This is the motivation for the definition below.
\begin{definition}\label{36}
Two distinct partitions $\varpi(H_1)$, $\varpi(H_2)$ are informationally independent, if the intersection of any block of $\varpi(H_1)$,  with any block of $\varpi(H_2)$, is non-empty:
\begin{eqnarray}
(\ket{v_1}+H_1)\cap (\ket{v_2}+H_2)\ne \emptyset;\;\;\;\ket {v_1}\in H_1^{\perp};\;\;\;\ket{v_2}\in H_2^{\perp}.
\end{eqnarray}
\end{definition}
If for some blocks $(\ket{v_1}+H_1)\cap (\ket{v_2}+H_2)= \emptyset$, then knowledge that a vector belongs to the block 
$\ket{v_1}+H_1$ of the partition $\varpi(H_1)$, implies that this vector does not belong to the block $\ket{v_2}+H_2$ of the partition $\varpi(H_2)$.
In this case the partitions $\varpi(H_1)$, $\varpi(H_2)$, are not informationally independent.

Partitions are isomorphic to subspaces, and
the concept of informationally independent partitions leads to the following definition of informationally independent subspaces.
\begin{definition}
Two distinct subspaces $H_1$, $H_2$ are informationally independent, if for any pair of vectors $\ket {v_1}\in H_1$ and $\ket {v_2}\in H_2$,
the intersection of the block $\ket{v_1}+H_1^{\perp}$ with the block $\ket{v_2}+H_2^{\perp}$, is non-empty:
\begin{eqnarray}\label{77}
(\ket{v_1}+H_1^{\perp})\cap (\ket{v_2}+H_2^{\perp})\ne \emptyset;\;\;\;\ket {v_1}\in H_1;\;\;\;\ket{v_2}\in H_2.
\end{eqnarray}
\end{definition}
The motivation for this definition in terms of quantum measurements, is as follows.
We asssume that for some blocks $(\ket{v_1}+H_1^{\perp})\cap (\ket{v_2}+H_2^{\perp})= \emptyset$, in which case 
if a state $\ket{s}$ belongs to the block $\ket{v_1}+H_1^{\perp}$, it cannot belong to the block $\ket{v_2}+H_2^{\perp}$. Then 
the measurement $\Pi(H_1)$ on $\ket{s}$ will collapse it
into ${\cal N}_1\ket{v_1}$, if the outcome is `yes'. But the measurement $\Pi(H_2)$ on $\ket{s}$, cannot collapse it into ${\cal N}_2\ket{v_2}$.
Therefore knowledge of the outcome of the $\Pi(H_1)$ measurement, gives information about the outcome of the $\Pi(H_2)$ measurement.
In this case, the $H_1$, $H_2$ are not informationally independent.

The following lemma will be used below to show that informational independence is the same concept as independence.

\begin{lemma}\label{pro2}
The following statements are equivalent:
\begin{itemize}
\item
The subspaces $H_1, H_2$ are informationally independent
\item
The $H_1, H_2$ are independent.
\item
The $H_1^{\perp}, H_2^{\perp}$ are a total set of subspaces.
\end{itemize}
\end{lemma}
\begin{proof}
We will prove the equivalence of the first two statements.
The equivalence of the last two statements has been proved earlier.

We assume that $H_1, H_2$ are independent, i.e., that $H_1\wedge H_2={\cal O}$.
For any vectors $\ket{v_1}\in H_1$ and $\ket{v_2}\in H_2$, we have to show that there exists vectors $\ket{a_1}\in H_1^{\perp}$ and $\ket{a_2}\in H_2^{\perp}$, such that
\begin{eqnarray}
\ket{v_1}+\ket{a_1}=\ket{v_2}+\ket{a_2}.
\end{eqnarray}
Because in this case Eq.(\ref{77}) holds.
In other words we have to show that there exists solution to the equation:
\begin{eqnarray}\label{60}
\ket{a_2}-\ket{a_1}=\ket{v_1}-\ket{v_2};\;\;\;\ket{v_1}-\ket{v_2}\in (H_1\vee H_2);\;\;\;\ket{a_2}-\ket{a_1}\in (H_1^{\perp} \vee H_2^{\perp}).
\end{eqnarray}
Here the unknowns are the vectors $\ket{a_2}, \ket{a_1}$.
We want to have a solution for all  vectors $\ket{v_1}\in H_1$ and $\ket{v_2}\in H_2$, and therefore the 
$\ket{v_1}-\ket{v_2}$ can be any vector in $H_1\vee H_2$.
In order to have a solution the $H_1\vee H_2$ should be a subspace of $H_1^{\perp} \vee H_2^{\perp}$.
If this is not the case then for $\ket{v_1}-\ket{v_2}$ in the set $(H_1\vee H_2)\setminus (H_1^{\perp} \vee H_2^{\perp})$, there exist
no solution.

The fact that  $H_1\wedge H_2={\cal O}$, implies that $H_1^{\perp} \vee H_2^{\perp}=H(d)$. 
Therefore in this case, $H_1\vee H_2\prec H_1^{\perp} \vee H_2^{\perp}=H(d)$, and
we can then find vectors $\ket{a_1}$, $\ket{a_2}$ which satisfy this equation.
This proves that $H_1, H_2$  are informationally independent.

Conversely, we assume that $H_1, H_2$ are informationally independent, i.e., that 
Eq.(\ref{60}) has a solution for all $\ket{v_1}-\ket{v_2}\in (H_1\vee H_2)$. Then
\begin{eqnarray}
H_1 \vee H_2\prec H_1^{\perp} \vee H_2^{\perp}=(H_1\wedge H_2)^{\perp}.
\end{eqnarray}
From this follows that
\begin{eqnarray}
H_1\wedge H_2\prec H_1 \vee H_2\prec (H_1\wedge H_2)^{\perp}.
\end{eqnarray}
This is possible only if $H_1\wedge H_2={\cal O}$.
Therefore the $H_1, H_2$ are independent.
\end{proof}

\subsection{Several informationally independent subspaces and measurements}

\begin{definition}\label{def290}
The subspaces $H_1,...,H_n$ of $H(d)$ are informationally independent, if for all $i=1,...,n$, the pairs of
subspaces ${\mathfrak H}_i$ and $H_i$ are 
informationally independent (according to the definition \ref{36}).
\end{definition}
\begin{proposition}\label{pro2}
The following statements are equivalent:
\begin{itemize}
\item
The subspaces $H_1, ...,H_n$ are informationally independent.
\item
The $H_1, ...,H_n$ are independent.
\item
The $H_1^{\perp}, ...,H_n^{\perp}$ are a strongly total set of subspaces.
\end{itemize}
\end{proposition}
\begin{proof}
We will prove the equivalence of the first two statements.
The equivalence of the last two statements has been proved earlier.

We assume that the $H_1, ...,H_n$ are independent, in which case according to
Eq.(\ref{33}) the 
${\mathfrak H}_i$ and $H_i$ are independent subspaces,
for all $i=1,...,n$.
Then lemma \ref{pro2} shows that they are also informationally independent subspaces.
Therefore according to the definition \ref{def290}, the $H_1,...,H_n$ are informationally independent subspaces.

This argument is also valid in the opposite direction, and it proves the converse.
\end{proof}

In view of this result we will use the simpler term independence for both independence and informational independence..

\section{Application: the pentagram in $H(3)$}\label{sec34}

An example of a formalism that requires a deeper understanding of the underlying concepts, is the pentagram which has been studied in the context of
contextuality \cite{C0,C1,C2,C3,C4,C5,C6,C7,C8}. 
In this section we apply our formalism to the pentagram.
\subsection{Background}

A context is a set ${\mathfrak C}=\{H_1,...,H_n\}$, of subspaces which commute pairwise ($H_i{\cal C} H_j$ for all $i,j$) or 
equivalently the corresponding projectors commute pairwise ($[\Pi(H_i),\Pi(H_j)]=0$ for all $i,j$).
The sublattice of ${\cal L}(d)$ generated by a context ${\mathfrak C}$ is distributive.
We consider two contexts 
\begin{eqnarray} 
{\mathfrak C}_1=\{H_1,H_2,..., H_n\};\;\;\;{\mathfrak C}_2=\{H_1, h_2,...,h_m\}.
\end{eqnarray} 
The subspaces in the first context commute pairwise, the subspaces in the second context commute pairwise,
but in general the $H_i$ does not commute with $h_j$. We will call them overlapping contexts because $H_1$ belongs to both of these contexts.

We consider the pentagram in $H(3)$, which has been studied originally in \cite{C4}.
To be specific, we consider the following states:
\begin{eqnarray}\label{g11}
\ket{s_0}=
\begin{pmatrix}
1\\
0\\
0\\
\end{pmatrix};\;\;\;
\ket{s_1}=\frac{1}{\sqrt{2}}
\begin{pmatrix}
0\\
1\\
1\\
\end{pmatrix};\;\;\;
\ket{s_2}=\frac{1}{\sqrt{3}}
\begin{pmatrix}
1\\
1\\
-1\\
\end{pmatrix};\;\;\;
\ket{s_3}=\frac{1}{\sqrt{6}}
\begin{pmatrix}
1\\
1\\
2\\
\end{pmatrix};\;\;\;
\ket{s_4}=\frac{1}{\sqrt{5}}
\begin{pmatrix}
0\\
-2\\
1\\
\end{pmatrix}.
\end{eqnarray}
The indices of these states belong to ${\mathbb Z}(5)$ (integers modulo $5$).
Any three of these vectors are independent.
It is easily seen that
\begin{eqnarray}\label{DFG}
\langle s_i\ket{s_{i+1}}=0;\;\;\;i\in {\mathbb Z}(5).
\end{eqnarray}
We call $H_i$ the one-dimensional subspace of $H(3)$, which contains the states $a\ket{s_i}$:
\begin{eqnarray}\label{vbn}
H_i=\{a\ket{s_i}\};\;\;\;i=0,...,4.
\end{eqnarray}
Their orthocomplements are the two-dimensional spaces 
\begin{eqnarray}\label{g12}
H_0^{\perp}=\left \{
\begin{pmatrix}
0\\
a\\
b\\
\end{pmatrix}\right \};\;\;
H_1^{\perp}=\left \{
\begin{pmatrix}
a\\
b\\
-b\\
\end{pmatrix}\right \};\;\;
H_2^{\perp}=\left \{
\begin{pmatrix}
a\\
b\\
a+b\\
\end{pmatrix}\right \};\;\;
H_3^{\perp}=\left \{
\begin{pmatrix}
2a\\
2b\\
-(a+b)\\
\end{pmatrix}\right \};\;\;
H_4^{\perp}=\left \{
\begin{pmatrix}
a\\
b\\
2b\\
\end{pmatrix}\right \}.
\end{eqnarray}
We also consider the projectors $\Pi(H_i)$ to the subspaces $H_i$:
\begin{eqnarray}
\Pi(H_i)=\ket{s_i}\bra{s_i};\;\;\;\Pi(H_i)\Pi(H_{i+1})=0;\;\;\;[\Pi(H_i), \Pi(H_{i+1})]=0;\;\;\;i\in {\mathbb Z}(5).
\end{eqnarray}
The terms exclusivity or local orthogonality are used in the literature, for the relation $\Pi(H_i)\Pi(H_{i+1})=0$ and its physical implications.

The ${\mathfrak C}_{i-1}=\{H_{i-1}, H_i\}$ and ${\mathfrak C}_i=\{H_i,H_{i+1}\}$ are overlapping contexts for all $i$,
and
\begin{eqnarray} 
[\Pi(H_{i-1}), \Pi(H_i)]=[\Pi(H_i), \Pi(H_{i+1})]=0;\;\;\;[\Pi(H_{i-1}), \Pi(H_{i+1})]\ne 0.
\end{eqnarray}

We perform measurements with the projector $\Pi(H_i)$ on an ensemble of states with density matrix $\rho$, and get the result `yes' with probability
$p_i={\rm Tr}[\rho \Pi(H_i)]$, and the result `no' with probability $1-p_i={\rm Tr}[\rho \Pi(H_i^{\perp})]$. 
We use the notation $1, 0$, for `yes' and `no' correspondingly.
The projectors $\Pi(H_i), \Pi(H_j)$ do not commute in general, and the corresponding probabilities will be measured using different ensembles of states described by the same density matrix $\rho$.
If $a_i=1,0$ is the outcome of the measurement $\Pi(H_i)$,
the distributions $p(a_i)$ and $p(a_i, a_{i+1})$ are measurable.

\subsection{A pentagram inequality within a non-contextual distributive hidden variable  theory}

In a non-contextual hidden variable theory, we assume that there exists a joint probability distribution
$p(a_0, a_1, a_2, a_3, a_4)$, for the outcomes of the measurements in the previous subsection, which has as marginals the measurable
distributions $p(a_i)$:
\begin{eqnarray}\label{cde}
p(a_i)=\sum  _{a_j\ne a_i}p(a_0, a_1, a_2, a_3, a_4).
\end{eqnarray}
We emphasize in this paper that this uses the law of total probability (proposition \ref{pro1q}), 
which is based on distributivity within the theory of Kolmogorov probabilities. 
So the non-contextual hidden variable theory, is assumed to be distributive,
and this is consistent with the classical nature of the hidden variable theory.

Since $\Pi(H_i)\Pi(H_{i+1})=0$,
the commuting measurements $\Pi(H_i), \Pi(H_{i+1})$ cannot both give $1$,
i.e., the $a_i, a_{i+1}$ cannot both be equal to 1. 
Consequently, if the $(a_0, a_1, a_2, a_3, a_4)$ has more than two `1', the probability $p(a_0, a_1, a_2, a_3, a_4)$ is zero.
This shows that the average number of `yes' answers, satisfies the inequality
\begin{eqnarray} 
\sum _{a_i} (a_0+a_1+a_2+a_3+a_4) p(a_0, a_1, a_2, a_3, a_4)\le 2.
\end{eqnarray}
Using Eq.(\ref{cde}) (which is based on distributivity), we rewrite this in terms of the marginal distributions, as
\begin{eqnarray} 
\sum _i a_ip(a_i)\le 2.
\end{eqnarray}
In the quantum language, this is
\begin{eqnarray} \label{bound}
\sum _{i=0}^4 {\rm Tr}[\rho \Pi(H_i)]\le 2,
\end{eqnarray}
where $\rho$ is the density matrix of the system.
It is known  \cite{C0,C1,C2,C3,C4,C5,C6,C7,C8} that quantum mechanics violates this inequality.
The left hand side can take the maximum value $\sqrt 5$, and this proves that 
there exists no joint probability distribution $p(a_0, a_1, a_2, a_3, a_4)$.
Quantum mechanics is a contextual theory where distributivity is replaced by the weaker property of modularity..

\subsection{Degree of independence in  the pentagram}\label{ex12}

We have explained above, that lack of distributivity makes problematic the use of marginal distributions in Eq.(\ref{cde}).
We now delve deeper into this, and we stress that 
in set theory there is a unique concept of disjointness, which leads to a clear concept of partition, and to the law of the total probability.
In the lattice of subspaces we have various levels of disjointness (independence),
and consequently, the meaning of the joint probability and  its marginals in Eq.(\ref{cde}), become problematic.
This is another way of expressing the problems associated with joint probabilities that involve non-commuting observables.
Only in the case of commuting observables, disjointness and pairwise disjointness are equivalent, the corresponding sublattice is distributive, and joint probabilities and their marginals are well defined.

We next show that in the pentagram we have the lowest level of disjointness (pairwise disjointness).
It is easily seen that any pair of the subspaces $H_0, H_1, H_2, H_3, H_4$ 
(given in Eq.(\ref{vbn})), are  independent.
Also ${\mathfrak H}_i=H(3)$ and ${\mathfrak H}_i^{\perp}={\cal O}$. Therefore
${\mathfrak H}_i^{\perp}\wedge H_i={\cal O}$ for all $i$, and
\begin{eqnarray}\label{ABC}
{\cal A}=\frac{1}{5}\sum _i\Pi({H}_i)=
\begin{pmatrix}
0.30&0.10&0\\
0.10&0.36&0.02\\
0&0.02&0.34\\
\end{pmatrix};\;\;\;\eta (\rho)={\rm Tr}({\cal A}\rho).
\end{eqnarray}
Our degree of independence $\eta (\rho)$, is the quantity in the inequality of Eq.(\ref{bound}) divided by $5$, for normalization purposes.
The $H_0, H_1, H_2, H_3, H_4$  are pairwise independent, which is the weakest form of independence.

The eigenvalues of the matrix ${\cal A}$ are
\begin{eqnarray}
\lambda _1=0.225;\;\;\;\lambda _2=0.338;\;\;\;\lambda _3=0.437.
\end{eqnarray}
Therefore $\eta (\rho)={\rm Tr}({\cal A}\rho)$ can reach the value $0.437$, which violates the inequality in Eq.(\ref{bound})
that has normalized upper bound $2/5$. In fact the value $0.437$
is only slightly lower than the maximum value ${\sqrt 5}/5=0.447$ given in the literature.

\subsection{A pentagram inequality within quantum  theory}
The following inequality gives an upper bound $2.5$ for $\sum {\rm Tr}[\rho \Pi(H_i)]$.
The same upper bound has also been given in \cite{UB} through a different argument. Here it is easily proved using the language of lattice theory.
This upper bound is of course higher than the maximum value for this quantity which is known to be ${\sqrt 5}$. 
\begin{proposition}
We consider the subspaces $H_i$ of the pentagram (an example of which 
is given in Eqs.(\ref{g11}), (\ref{vbn})). If $\rho$ is the density matrix of the system, then
\begin{eqnarray} 
\sum _{i=0}^4 {\rm Tr}[\rho \Pi(H_i)]\le 2.5
\end{eqnarray}
\end{proposition}
\begin{proof}
The fact that $\Pi(H_i)\Pi(H_{i+1})=0$, implies that $H_i\prec H_{i+1}^{\perp}$.
Consequently,
\begin{eqnarray} 
{\rm Tr}[\rho \Pi(H_i)]\le {\rm Tr}[\rho \Pi(H_{i+1}^{\perp})].
\end{eqnarray}
We use these inequalities for all values $i$, and adding them we prove that
\begin{eqnarray} 
\sum _{i=0}^4 {\rm Tr}[\rho \Pi(H_i)]\le \sum _{i=0}^4 {\rm Tr}[\rho \Pi(H_i^{\perp})].
\end{eqnarray}
But we also have
\begin{eqnarray} 
\sum _{i=0}^4 {\rm Tr}[\rho \Pi(H_i)]+\sum _{i=0}^4 {\rm Tr}[\rho \Pi(H_i^{\perp})]=5.
\end{eqnarray}
From the last two relations, follows the inequality in the proposition.
\end{proof}

\section{Discussion}

An important property in classical physics and classical logic, formalized with Boolean algebra, is distributivity.
In quantum physics and quantum logic, formalized with the Birkhoff-von Neumann lattice of subspaces, it is replaced by the weaker property of modularity (in systems with finite-dimensional Hilbert space). This has profound implications, some of which are discussed in this paper.
Of course, within the lattice ${\cal L}(d)$ there are sublattices which are distributive (e.g., when the subspaces commute), and in those `islands' results similar to classical physics do hold.

Within the formalism of phase space methods, we have considered the subspaces $H_1,...,H_n$ of $H(d)$, and the quasi-probability distributions $R(i)$ in Eq.(\ref{AA}).
We also introduced the quasi-probability distributions ${\widetilde R}(i)$ and ${\widehat R}(i)$ in Eq.(\ref{AAA}), and discussed their physical meaning in terms of measurements.

In this general context, we have introduced the concepts of independence and totalness.
We have shown that in quantum theory there are many levels of independence,
from pairwise independence up to independence, and they are quantified with the degree of independence that compares the distributions
$R(i)$ and ${\widetilde R}(i)$.
There are also many levels of totalness, 
from pairwise totalness  up to totalness, and they are quantified with the degree of totalness that compares the distributions
$R(i)$ and ${\widehat R}(i)$.
The existence of various levels of independence and totalness, is intimately related to the lack of distributivity in quantum theory.
In set theory where distributivity holds, there is a single concept of independence and a single concept of totalness.

There is a duality between a set of independent subspaces, and the total set of 
the orthocomplements of these subspaces.
Orthocomplementation (logical NOT operation) transforms independence into totalness.

One application of these ideas, is the law of total probability which is used to define marginals of probability distributions.
We have explained that its proof relies on the distributivity property, and its application in non-distributive structures is problematic.
This has been used in the pentagram, where a non-contextual distributive hidden variable theory leads to the inequality in Eq.(\ref{bound}), which is violated by quantum mechanics.  

The work studies quantum theory from the angle of non-distributivity.
It introduces novel concepts like the various levels of independence and the various levels of totalness,
which can play a complementary role to non-commutativity, for the description of quantum phenomena.

\end{document}